\documentclass[journal,romanappendices]{IEEEtran} 

\usepackage[tight,footnotesize]{subfigure}

\usepackage{graphicx}
\usepackage{amssymb}
\usepackage{array}
\usepackage{dsfont} 

\usepackage{cite}
\usepackage{amsmath}
\usepackage{amsthm}
\usepackage{graphicx}
\usepackage{amssymb}

\usepackage{array}
\usepackage{color}

\usepackage{multirow}
\usepackage{ifpdf}
\usepackage[T1]{fontenc}
\graphicspath{{./Figs/}}                   

\theoremstyle{definition}
\newtheorem{definition}{Definition}
\newtheorem{lemma}{Lemma}
\newtheorem{theorem}{Theorem}
\newtheorem{corollary}{Corollary}

\newcommand{\prob}{\mathbb{P}}

\newcommand{\Ex}{\mathbb{E}}
\newcommand{\Var}{\mathbb{V}\mathrm{ar}}
\newcommand{\ve}[1]{\boldsymbol{#1}}
\newcommand{\R}{\mathbb{R}}
\newcommand{\C}{\mathbb{C}}

\newcommand{\snr}{\text{SNR}}

\newcommand{\tr}{\text{tr}}

\newcommand{\defeq}{\triangleq}

\newcommand{\diag}{\text{diag}}

\newcommand{\Hip}{\mathcal{H}}

\newcommand{\CN}{\mathcal{CN}}

\newcommand{\clust}{\mathcal{C}}
\begin{document}
\title{Exploiting Spatial Correlation in Energy Constrained Distributed Detection}

\author{{Juan Augusto Maya,~\IEEEmembership{Student Member,~IEEE,} Cecilia G. Galarza and Leonardo Rey Vega,~\IEEEmembership{Member,~IEEE.}}

\thanks{The authors are with the University of Buenos Aires (UBA), Buenos Aires, Argentina. L. Rey Vega and C. Galarza are also with the CSC-CONICET, Buenos Aires, Argentina. 
Emails:  \{jmaya, lrey, cgalar\}@fi.uba.ar} %
\thanks{This work was partially supported by the Peruilh grant (UBA), and the projects UBACYT 20020130100751BA and CONICET PIP 112 20110100997.} %
}



\maketitle

\begin{abstract}
We consider the detection of a correlated random process immersed in noise in a wireless sensor network. Each node has an individual energy constraint and the communication with the processing central units are affected by the path loss propagation effect. Guided by energy efficiency concerns, we consider the partition of the whole network into clusters, each one with a coordination node or \emph{cluster head}. Thus, the nodes transmit their measurements to the corresponding cluster heads, which after some processing, communicate a summary of the received information to the fusion center, which takes the final decision about the state of the nature. As the network has a fixed size, communication within smaller clusters will be less affected by the path loss effect, reducing energy consumption in the information exchange process between nodes and cluster heads.  However, this limits the capability of the network of beneficially exploiting  the spatial correlation of the process, specially when the spatial correlation coherence of the process is of the same scale as the clusters size. Therefore, a trade-off is established between the energy efficiency and the beneficial use of spatial correlation.  The study of this trade-off is the main goal of this paper.  We derive tight approximations of the false alarm and miss-detection error probabilities under the Neyman-Pearson framework for the above scenario. We also consider the application of these results to a particular network and correlation model  obtaining closed form expressions. Finally, we validate the results for more general network and correlation models through numerical simulations. 

\end{abstract}

\begin{IEEEkeywords}
Distributed detection; correlated measurements; energy and bandwidth constraints; wireless sensor networks; hierarchical clustering; spectrum sensing; cognitive radio.   
\end{IEEEkeywords}
\IEEEpeerreviewmaketitle

\section{Introduction}

\IEEEPARstart{D}{istributed}
detection based on wireless sensor networks (WSN) is a topic which has attracted great interest in recent years (see \cite{veeravalli-varshney-2012} and references therein). A typical WSN has a large number of sensor nodes which are generally low-cost battery-powered devices with limited sensing, computing, and communication capabilities. Sensors acquire noisy measurements, perform simple data processing and propagate the information into the WSN to reach a decision about a physical phenomenon occurring in the coverage area. 

Network resources, such as energy and bandwidth, are scarce, expensive, and key variables when the design is focused on processing latency and detection performance. Clever data processing strategies are required to maximize performance under resource constraints. Appropriate choices of the processing strategy could largely impact on the total cost or on the life cycle of the WSN, and as such, make its deployment on remote locations economically viable or not. 
       
\subsection{Related work}
Distributed detection theory has been much studied in the past. Starting with the seminal work of Tenney and Sandell \cite{Tenney_1981}, several results have been derived on how each node process the available information and communicates with the \emph{fusion center} (FC) where the final decision on the true state of nature is taken. Under this setup, digital transmission schemes, where appropriate quantization rules have to be designed have been considered in \cite{VarshneyDistDet, TsitsiklisDescDet, ChamberlandVeeravalli2003} (see also references therein). On the other hand, analog communication schemes were also studied in the past (see for example \cite{Sayeed2003DistClassifClusters, Veeravalli_2003, ChamberlandVeeravalli2004,LiDai_DetSignalMAC,CohenLLRMAC,Bianchi_2011}).


Instead of independent observations, often assumed in the previously mentioned works, 
we explore a more realistic scenario with spatial dependence measurements. Distributed detection problems with correlated observations can be considerably more challenging than their conditionally independent counterparts \cite{willett2000good}. Moreover, new theoretical frameworks have been proposed recently \cite{Varshney2012NewFramework, Maya_2015}.

The design of distributed processing strategies exploiting spatial and/or temporal correlation among data is, in general, an open problem. It is well known that signal correlation can help to improve the detection performance, specially when low quality sensor measurements are available \cite{ChamberlandVeeravalli2004,PoorLattice2D,Veeravali2008CoopSensing}. Clearly, a clever use of the correlation in space and/or time requires cooperation among nodes. This problem has also been studied in previous works \cite{Sung_2007,Tong_2007, Varshney2012FusingData,Veeravali2008CoopSensing}. Additional studies have explored the effects of correlation on the performance of distributed detection systems \cite{Aalo1989_DDCorrSensors,Drakopoulos1991_CorrDecisions}. 
 
An important application that has gained a lot of attention in the recent years under this scenario is spectrum sensing for cognitive radio \cite{yucek2009survey,2015PoorSpectrumSensing}. Here, unlicensed, or secondary users (SU), want to detect the presence or absence of the licensed, or primary users (PU), in order to use the spectrum when it is available. SUs could cooperate and build a WSN to achieve this task. When a PU uses the spectrum, it transmits a signal that may be sensed by the WSN. Due to the shadowing propagation effect, different close nodes sense correlated measurements. Thus, dependent observations must be considered in order to develop reliable spectrum sensing schemes. Moreover, the signal model and the scheme developed in this work can be directly applied to a cognitive radio scenario \cite{Veeravali2008CoopSensing}.         

Typically, the nodes of a WSN are distributed on a vast area to sense the whole region of interest. Thus, the energy needed for reliable communication increases with the transmision distance due to the effect of path loss. In order to save transmission energy, it is convenient to reduce the region where sensors transmit high precision data. In this way, a hierarchical WSN is introduced: sensors are grouped in clusters with a designated cluster head (CH). Each CH processes the data independently of the other ones and communicates a summary of the information gathered by its nodes to the FC, where the final decision about the state of the nature is made. Previous works \cite{Sayeed2003DistClassifClusters}, \cite{Varshney2005distributed}  have already considered the idea of partitioning the region to be sensed into clusters. However, in \cite{Varshney2005distributed} the authors considered conditionally independent and identically distributed measurements, and in \cite{Sayeed2003DistClassifClusters}, spatially coherence regions were assumed, but dependence among them is not taken into account. 


\subsection{Contributions}
In this paper, we develop a clustering WSN model with energy and bandwidth constraints. We assume that each cluster processes the information independently of other clusters. Moreover, we consider that each individual cluster only knows the local correlation structure of the process. This hypothesis simplifies the scheme in the sense that communication of high data rate information between CHs or nodes from different clusters is avoided. As these exchanges would require to overcome path loss effects, an energy saving in each node will be obtained. However, in a WSN with correlated measurements, this strategy limits the ability of the network of beneficially exploiting the correlation of the process. Therefore, a trade-off is established between the beneficial use of the underlying random process spatial correlation and the energy spent by each node in the network. This is a trade-off between the energy required to reliably communicate each node measurement and the ability to effectively use the cooperation scheme required for exploiting the spatial correlation information.

To characterize this trade-off, we derive tight approximations of the false alarm and miss-detection error probabilities for statistics built with sum of independent although not necessarily identically distributed random variables. The expressions involve the sum of the logarithmic moment generating function (LMGF) of each random variable and its first and second derivatives. For example, in this paper we consider Gaussian processes for which the statistic is a quadratic form of a Gaussian correlated random vector. In this case, the statistic can be expressed as a weighted sum of i.i.d. random variables and the error probabilities can be easily computed.   

We obtain closed forms expressions for a particular model of the sensed process. Also, we numerically  compute the average performance for more general processes and when the sensors are randomly placed in space and show that similar conclusions can be drawn. These conclusions are summarized as follows: 
\begin{enumerate}
\item In a WSN with Gaussian measurements there are two clearly different operating regimes: an energy limited regime, where the energy dedicated to transmit the data limits the performance, and a correlation-limited regime, where the performance is limited by the poor use of the spatial correlation of the process. The ranges of validity of each of these regimes are given by different parameters of the WSN some of which are under control of the designer.  The characterization of this problem gives crucial information for the deployment of a WSN.  
\item In some cases of interest the use of local correlation information is almost as good as the use of the global correlation statistics of the sensed processes. This is a very important issue from an implementation point of view, as the use of global correlation information requires a greater degree of complexity in the network.
\item The use of local correlation information becomes more profitable when the sensor measurements are less reliable (i.e., the measurement signal to noise ratio is low).
 \end{enumerate}

\subsection{Notation and Organization}
Scalars are written in italic, vectors in boldface and matrices in capital letters. $A_n$ and $B_{nm}$ are  square and  rectangular matrices of sizes $n\times n$ and $n\times m$, respectively. 
We define column vectors with components $a_k$ as $\ve a =[a_k]_{k\in\mathcal{A}}$, where $k$ belongs to the set of indexes $\mathcal{A}$. Similarly, we denote matrices  $A_{nm}=[\ve{a}_{k}]_{k\in\mathcal{A}}$ built through the horizontal concatenation of the column vectors $\{\ve{a}_{k}\}$, where $\ve a_k\in\C^{n\times1}$ and the cardinality of $\mathcal{A}$ is $|\mathcal{A}|=m$.
$\det (A_n)$ and $\tr (A_n)$ are the determinant and the trace of $A_n$, respectively. $(\cdot)^T$ and $(\cdot)^H$ denote transpose and transpose conjugate. 
We do not distinguish between random variables and their realization values when the context is clear. $p(\ve{x}|\Hip)$ is the probability density function (p.d.f.) of $\ve{x}$ under $\Hip$. $\prob_i(\cdot)$, and $\Ex_i(\cdot)$ denote probability and the expectation respectively, both computed under hypothesis $\Hip_i$.
Given two statistics $T$ and $T'$, we denote that they have the same error probabilities writing $T\equiv T'$. 
Given a function $\mu(\cdot)$, its first and second derivatives are denoted as $\dot{\mu}(\cdot)$ and $\ddot{\mu}(\cdot)$. The notation $f(n)= \mathcal{O}\left( g(n)\right)$ means that $\lim_{n\rightarrow \infty} f(n)/g(n) = c$, where $c$ is a non-zero constant. 

The paper is organized as follows. In Section \ref{sec:model}, we present the WSN model and the binary hypothesis testing problem. In Section \ref{sec:precoding}, we describe the precoding strategies used for stationary and non-stationary random processes. In Section \ref{sec:globalStats}, we derive a new global statistic and formulate the main problem of the paper. In Section \ref{sec:errorProb}, we compute tight approximations of the error probabilities to be used in Section \ref{sec:performance} to analyse the performance of the proposed strategy. Finally, in Section \ref{sec:conclusions} we elaborate on the main conclusions. Technical proofs are provided in the appendices.
 
\section{WSN Model}
\label{sec:model}
We analyse a sensor network with three hierarchies: nodes, clusters with their corresponding CH, and the FC. The nodes are grouped in clusters. Each CH gathers information from the nodes in its corresponding cluster, performs some processing independently of other clusters and transmit a summary to the FC, where the final decision is made. A sketch of the WSN is shown in Fig. \ref{fig:generalWSN}.

\begin{figure}[hbt]
\centering
\includegraphics[width=1\linewidth]{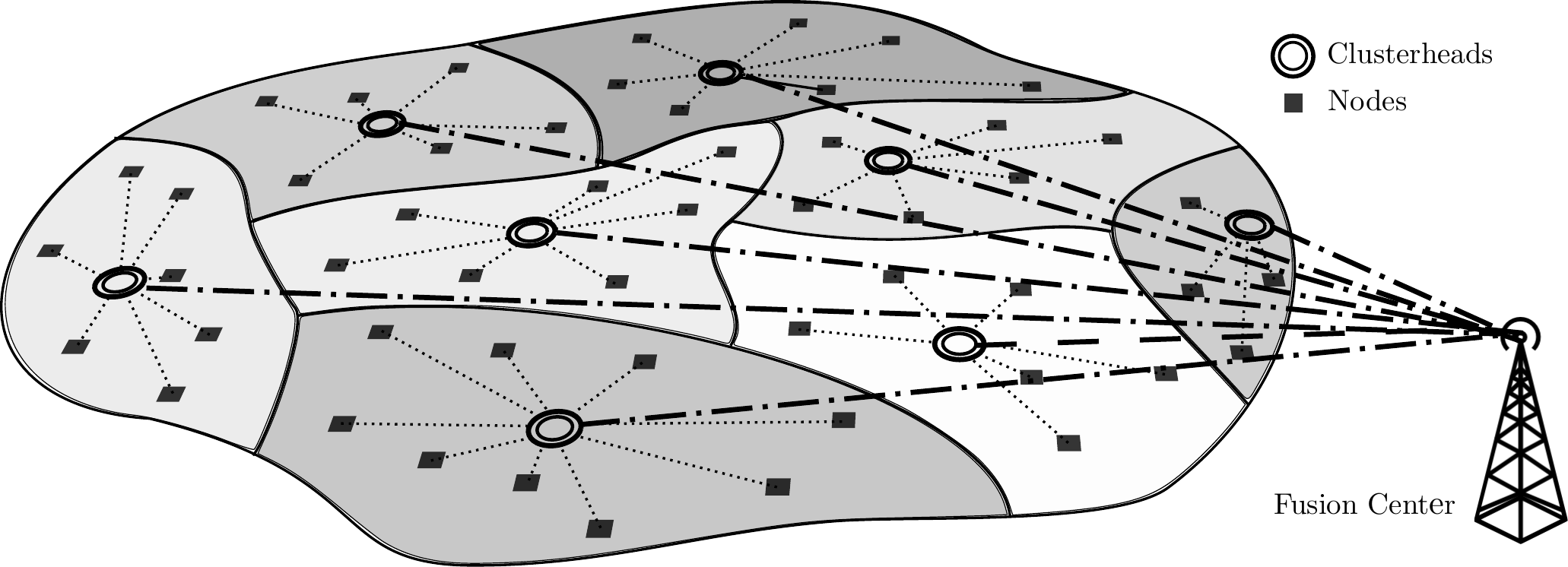}
\caption{A hierarchical WSN, where the information is propagated from the nodes to the CHs, and from them to the FC, which decides about the state of the nature.}
\label{fig:generalWSN}
\end{figure}
The WSN has $m$ clusters and a total amount of $n$ nodes which are arbitrary distributed in an area $nA_0$, which for simplicity it will be assumed to be a square and $A_0$ a constant value (with $m^2$ units). Each cluster has $l_i$ nodes, $i=1,\dots,m$, such that $n=\sum_{i=1}^m l_i$. Let $\clust_i$ be the set of sensors of the $i$-th cluster, $i=1,\dots,m$. $\{\clust_i\}_{i=1}^m$ is a partition of the set of sensors $[1:n]$, i.e., $\clust_i\cap\clust_j=\emptyset\ \forall i \neq j$, $i=1,\dots,m$, and $\cup_{i=1}^m \clust_i=[1:n]$. In the following, we provide a description on the main elements of the general WSN model we will work with.
 
The $k$-th sensor of the WSN takes a measurement described under each hypothesis by
\begin{equation}
\left\{
\begin{array}{llll}
\Hip_1: & y_k&=s_k + v_k,& \\
\Hip_0: & y_k&=v_k,& k=1,\dots,n,
\end{array}
\right.
\label{eq:GenMod}
\end{equation} 
where $s_k$ is a correlated zero-mean circularly-symmetric complex Gaussian process with variance $\sigma^2_s$ and autocorrelation function (ACF) $\mathbb{E}[s_js_k]=R_s(j,k)$ with $j,k\in[1:n]$. In general, $\{s_k\}$ is a non-stationary spatial process.  
We also assume that $v_k$ is a zero-mean circularly-symmetric complex white Gaussian noise independent of $s_k$ with variance $\sigma_v^2$. Thus, $y_k$ is Gaussian distributed either under $\Hip_0$ or $\Hip_1$. The above model can be easily generalized to the case in which each sensor takes several measurements at different time instants, and the processes have temporal correlation. All the mathematical developments and conclusions in this paper can be extended to that case. However, in order to keep things simple we consider only the case in which each sensor takes only one measurement. This is in line with several works on distributed detection on WSN \cite{ChamberlandVeeravalli2004,LiDai_DetSignalMAC,CohenLLRMAC}.

Define the vector $\ve{y}_i= [y_k]_{k\in\clust_i}$ as the set of measurements taken by the nodes in the $i$-th cluster. Define $\ve{s}_i= [s_k]_{k\in\clust_i}$ and $\ve{v}_i= [v_k]_{k\in\clust_i}$ accordingly. The covariance matrices of $\ve{s}_i$ and $\ve{v}_i$ are, respectively, $\Sigma_{s,l_i}$ and $\sigma_v^2 I_{l_i}$ where the elements of $\Sigma_{s,l_i}$ are $R_s(j,k)$, with $j,k\in \clust_i$, and $I_{l_i}$ is the identity matrix of dimension $l_i$. Therefore, the covariance matrix of $\ve{y}_i$ under $\Hip_0$ is $\Sigma_{0,l_i}=\sigma_v^2 I_{l_i}$, and under $\Hip_1$ is $\Sigma_{1,l_i} = \Sigma_{s,l_i}+\sigma_v^2 I_{l}$. Let $\ve{y}=[\ve{y}_i]_{i\in[1:m]}$ be the measurements taken by all nodes in the network and define $\ve{s}=[\ve{s}_i]_{i\in[1:m]}$, $\ve{v}=[\ve{v}_i]_{i\in[1:m]}$. The covariance matrices of $\ve{s}$ and $\ve{v}$ are, respectively, $\Sigma_{s,n}$ and $\sigma_v^2 I_{n}$ where the elements of $\Sigma_{s,n}$ are $(\Sigma_{s,n})_{j,k}=R_s(j,k)$. The covariance matrix of $\ve{y}$ under $\Hip_0$ and $\Hip_1$, are, respectively, 
\begin{equation}
\Sigma_{0,n}=\sigma_v^2 I_{n}, \qquad \Sigma_{1,n} = \Sigma_{s,n}+\sigma_v^2 I_{n}. \label{eq:covy}
\end{equation}
Notice that $\Sigma_{0,n}$ and $\Sigma_{1,n}$ can be thought as $m\times m$ block matrices, where the diagonal blocks are $\Sigma_{0,l_i}$ and $\Sigma_{1,l_i}$, $i=1,\dots,m$.

\section{Cluster Precoding Strategies}
\label{sec:precoding}
In this work, we consider that the communication between the nodes and their corresponding CH is through either a \emph{parallel-access channel} (PAC) or multiple channel uses\footnote{We say that a channel use takes place when a symbol is transmitted.} of a \emph{multiple-access channel} (MAC). In the case of the PAC, the amount of channel uses coincides with the number of available orthogonal channels in each cluster. In both cases, we assume that the nodes communicate with its corresponding CH without any interference from other clusters. Consider now that the $i$-th cluster has $l_i'$ channel uses available for transmission, with $l_i'\leq l_i$. Each channel use is associated to either an orthogonal time slot or a frequency band. Then, the processing strategy may use $l_i'$ time slots, $l_i'$ frequency bands or a combination of them such that the product of time slots and frequency bands is $l_i'$. An illustration of the distributed scheme is shown in Fig. \ref{fig:Nodes2CH}. 
\begin{figure}[hbt]
\centering
\subfigure[\label{fig:Nodes2CH}]{\includegraphics[width=\linewidth]{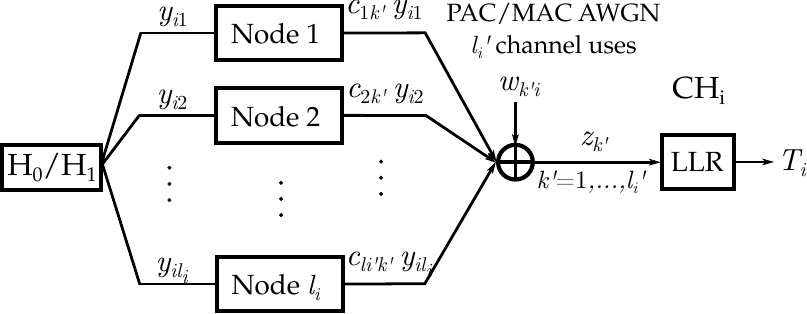}}
\vspace*{3mm}
\subfigure[\label{fig:CH2FC}]{\includegraphics[width=.6\linewidth]{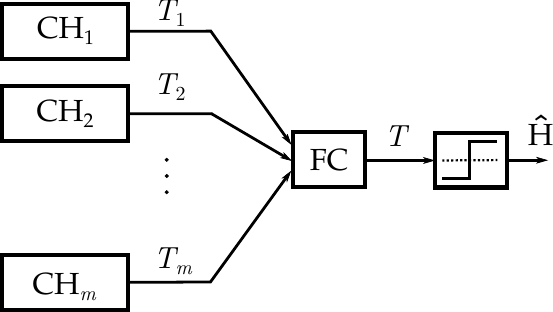}}
\caption{A distributed scheme for gathering information in a WSN. (a) The nodes use a PAC or a MAC channel to communicate with their corresponding CH. (b) The FC collects the processed information of each CH and takes a decision.}
\label{fig:wsnScheme}
\end{figure}

\subsection{The PAC precoding strategy}
If a PAC is considered, the received signal at the $i$-th CH is 
\begin{align}
z_{k}&= a_{k}\, c_{k}\, y_k + w_{k},\ \ k\in \mathcal{U}_i, \label{eq:zCH-PAC-1}
\end{align}
where $w_{k}$ is the zero-mean communication white noise with variance $\sigma_w^2$ and circularly-symmetric complex Gaussian distribution independent of the measurements $y_k$, $a_{k}=\|\ve x_k - \ve x_{\text{CH}_i} \|^{-\frac{\epsilon}{2}}$ is the factor associated with the path loss effect between the position of the node $\ve x_k\in \R^2$ and the position of the corresponding cluster head $\ve x_{\text{CH}_i}\in \R^2$, $\epsilon$ is the path loss exponent, $c_{k}$ is the precoding coefficient used by the $k$-th sensor, and $\mathcal{U}_i\subseteq \mathcal{C}_i$ is the set of channel uses indexes for the $i$-th cluster with cardinality $|\mathcal{U}_i|=l_i'$. Under this strategy, each sensor has a single channel use to communicate its measurement. 
If only $l_i'< l_i$ channel uses are available, $l_i-l_i'$ sensors in the cluster stay silent. 
Let $C_i=\diag((c_{k})_{k\in\mathcal{U}_i})$ be the \emph{precoding matrix} of the $i$-th cluster, $A_i=\diag((a_k)_{k\in \mathcal{U}_i})$, and define the column vectors $\ve{z_i}= [z_{k}]_{k\in\mathcal{U}_i}$, and $\ve{w}_i=[w_{k}]_{k\in\clust_i'}$ with covariance matrix $\sigma^2_w I_{l_i'}$. The received signal vector in each cluster can be expressed as
\begin{equation}
\label{eq:zCH}
\ve z_i = A_i\,C_i \ve{y}_i + \ve{w}_i, \ i=1,\dots, m.
\end{equation}
with covariance matrix under $\Hip_j$, $j=0,1$ given by,
\begin{align}
\Xi_{j,l_i}&=C_i^H A_i\Sigma_{j,l_i}A_iC_i+\sigma^2_w I_{l_i'}\label{eq:Xij_l}.
\end{align}

The average energy consumed by the $k$-th node in the corresponding $i$-th cluster during one of the $l_i'$ scheduled transmissions is:
\begin{align}
E_{k} = \mathbb{E}\left(|c_{k}y_k|^2 \right) = |c_{k}|^2(\sigma_v^2+p_1\sigma_s^2),
\label{eq:energyAvg}
\end{align}
where $p_1$ is the \emph{a priori} probability of the state of nature $\Hip_1$. 
When $p_1$ is unknown, a natural upper-bound for $E_k$ is obtained by taking $p_1=1$. In this work, we will consider the following individual sensor energy constraint:
\begin{equation}
E_k \leq \bar{E}\ \forall k \in [1:n].  
\label{eq:SensorEC}
\end{equation}

For the particular case when the sensed process is uncorrelated, and considering an individual sensor energy constraint, it can be easily proved that the optimal strategy consists in using all the available energy at each sensor. 
In this paper, we will consider this strategy even for correlated processes. The precoding strategy for a PAC channel named amplify and forward strategy (AFS-PAC) is defined next. 
\begin{definition}[AFS-PAC]
\label{PS-PAC}
Considering the $i$-th cluster, (\ref{eq:energyAvg}) and (\ref{eq:SensorEC}), the  precoding matrix for the AFS-PAC strategy is
\begin{equation}
C_i=\sqrt{\tfrac{\bar{E}}{\sigma_s^2+\sigma_v^2}} I_{l_i'}, \ i=1,\dots, m.
\label{eq:Prec-AFS-PAC}
\end{equation}
\end{definition}

\subsection{The MAC precoding strategies}
In the case that a MAC is used, each sensor communicates with its corresponding CH using different gains for each channel use. The signal collected at the $i$-th CH is a noisy version of the coherent superposition of the symbols transmitted by the $l_i$ sensors through the corresponding MAC:
\begin{align}
z_{k'}&=\sum_{k\in\clust_i} a_{k} c_{k k'} y_k + w_{k'},\ \ k'\in \mathcal{U}_i, \label{eq:zCH-MAC-1}
\end{align}
where $a_{k}$, $ w_{k'}$ and $\mathcal{U}_i$ are the same quantities defined above. Thus, the $k$-th sensor transmits in the $k'$-th channel use its measurement $y_k$ scaled by the precoding coefficient $c_{k k'}$ through the MAC channel. Stacking the measurements from each cluster, the same vector signal model (\ref{eq:zCH}) is obtained, although for the MAC channel, the precoding matrix defined by $C_i = [c_{kk'}]_{k\in\mathcal{C}_i,k'\in\mathcal{U}_i}$ is not, in general, a diagonal matrix. Considering (\ref{eq:zCH-MAC-1}), the average energy consumed by the $k$-th node in the corresponding $i$-th cluster during the $l_i'$ scheduled MAC channel uses is:
\begin{align}
E_{k} &= \mathbb{E}\left(\sum_{k'\in\mathcal{U}_i} |c_{kk'}y_k|^2 \right)=(\sigma_v^2+p_1\sigma_s^2)\sum_{k'\in\mathcal{U}_i} |c_{kk'}|^2.
\label{eq:energyAvg-MAC}
\end{align}



If $\{s_k\}$ is spatial-stationary, two schemes presented in \cite{Maya_2015} may be used to optimally allocate the energy and bandwidth in successive uses of the MAC channel. Both schemes, called \emph{principal components strategy in a multiple-access channel} (PCS-MAC) and \emph{principal frequency strategy in a multiple-access channel} (PFS-MAC), are asymptotically the best strategies among the orthogonal schemes considering energy and bandwidth constraints and the error exponents as metrics of performance. Although error exponents are sometimes useful to compare and define strategies of large networks, they do not provide an exact characterization of the error probabilities, which are the typically required  metrics of performance for WSNs. Therefore, in this work, we mainly use the error probabilities as metrics of performance. 
Because the optimum strategy for a MAC, in terms of error probabilities, remains an open problem, we will use the PFS-MAC strategy (which is optimal in terms of error exponents) and show that good results are obtained. 

On the other hand, 
if $\{s_k\}$ is not spatial-stationary, in general, the covariance matrix of the process does not have a regular structure (i.e. Toeplitz form). However, in this case the PCS-MAC scheme presented in \cite[Def. 2]{Maya_2015} may be used to optimally allocate the transmission energy and bandwidth. This scheme requires a built infrastructure that might not available be available on any WSN. For that reason, in the spatial non-stationary case, we will consider the AFS-PAC strategy.


Next, we describe the PFS-MAC strategy for stationary processes. During each channel use, each CH receives a noisy version of a given frequency bin of the discrete Fourier transform (DFT) of the corresponding measurement vector $\ve y_i$. The idea behind this strategy is to transmit to the corresponding CH the \emph{most distinguishing} spectral components of the power spectral density (PSD) of the sensed process $\phi(\nu)$ (achieved when $l$ grows unbounded).
In particular, we will use for each cluster a simple almost-optimal solution that employs the PFS-MAC scheme with a constant energy profile (EP) on a given set of transmission frequencies.  The strategy, called ON/OFF-EP (see \cite{Maya_2015}), is defined as follows:
\begin{definition}[PFS-MAC with ON/OFF-EP]
\label{def:PFS-MAC}
Consider the $i$-th cluster. For a given $l_i$, denote by $(j_1,j_2,\dots,j_{l_i})$ a permutation of $\{1,2,\dots,l_i\}$ such that $\phi\left(\frac{j_1}{l_i}\right) \geq\phi\left(\frac{j_2}{l_i}\right)\geq \cdots \geq \phi\left(\frac{j_{l_i}}{l_i}\right)$. Let $l_i'$ be the number of channel uses or  bandwidth constraint. Let $\beta_i=\frac{l_i'}{l_i}$ be the fraction of channel uses (or \emph{bandwidth}) used at cluster $i$. Then, the set of channel uses is defined by $\mathcal{U}_i= \{j_1,\dots,j_{l_i'} \}$ and the precoding matrix of the PFS-MAC strategy with ON/OFF-EP is
\begin{equation}
C_i=\gamma_i A^{-1}_{i} F_{l_i l_i'}, \ i=1,\dots,m,
\label{PFS-matrix-pD}
\end{equation}
where $F_{l_i l_i'}=[\ve{\zeta}_{j_1},\dots,\ve{\zeta}_{j_{l_i'}}]\in\C^{l_i\times l_i'}$ is a sub-matrix of the normalized DFT matrix of size $l_i\times l_i$, i.e.,  
$\ve{\zeta}_{k'}=[f_{1k'},f_{2k'},\dots, f_{l_i k'}]^T$ with $f_{kk'} = \frac{1}{\sqrt{l_i}}\exp\left(\jmath 2\pi \frac{(k-1) (k'-1)}{l_i}\right)$, $k=1,\dots,l_i$; the diagonal matrix $A^{-1}_{i}$ performs a channel  inversion to compensate for the different gains introduced by the path loss. Using (\ref{PFS-matrix-pD}) in (\ref{eq:energyAvg-MAC}), and considering the individual sensor energy constraint (\ref{eq:SensorEC}), we obtain that $\gamma_i$ is limited by the longest distance between the nodes and the CH of the $i$-th cluster. Then, 
\begin{equation}
\gamma_i= \sqrt{\tfrac{\bar{E}}{(\sigma_s^2+\sigma_v^2)\beta_i (d_{\text{max},i})^\epsilon}}, \ i=1,\dots, m,\label{eq:Prec-PFS-MAC}
\end{equation}
where $d_{\text{max},i} = \max_{k\in\clust_i} \|\ve x_k-\ve x_{\text{CH}_i}\|$.
\end{definition}

Using the above scheme, the covariance matrices of $\ve{z}_i$ under $\Hip_1$ and $\Hip_0$ are, respectively,
\begin{align}
\label{eq:cov_zlp1}
\Xi_{1,l_i'}&= \gamma_{i}^2 F_{l_i l_i'}^H\Sigma_{s,l_i} F_{l_i l_i'} + (\gamma_{i}^2\sigma^2_v + \sigma^2_w) I_{l_i'}\\
\Xi_{0,l_i'}&=(\gamma_i^2 \sigma^2_v + \sigma^2_w) I_{l_i'}.
\end{align}

\section{Detection Statistics}
\label{sec:globalStats}

\subsection{Full Correlation Strategy}
Suppose that the FC has direct access to all the measurements available in the CHs $\ve z =[\ve z_1^T,\dots,\ve z_m^T]^T$ through ideal noiseless channels. The covariance matrix of $\ve z$ under $\Hip_j$, $j=0,1$, is 
\begin{equation}
\Xi_{i,n'} = B_{nn'}^H\Sigma_{j,n} B_{nn'} + \sigma^2_w I_{n'},\label{eq:Xi-np}
\end{equation}
where $B_{nn'} = \diag(A_1 C_1, \dots, A_m C_m)$. 
Consider the Neyman-Pearson problem \cite{Kay_SSP} for a fixed false alarm probability level $\alpha$, where a false alarm event occurs when $\Hip_1$ is declared but $\Hip_0$ is true.
Let the logarithmic likelihood ratio (LLR) \cite{Kay_SSP} of the full correlation strategy (FCS) be:
\begin{align}
T_\text{FCS}(\ve{z}) &= \log \frac{p(\ve{z}|\Hip_1)}{p(\ve{z}|\Hip_0)}
\equiv \ve{z}^H \left( \Xi_{0,n'}^{-1}-\Xi_{1,n'}^{-1} \right)\ve{z},
\label{eq:FCS}
\end{align}
where  we discarded the constant term $\log\det(\Xi_{0,n'}\Xi_{1,n'}^{-1})$ without affecting the performance of the statistic. Now, under the Neyman-Pearson setting, the optimal decision rule chooses $\Hip_1$ if $T_\text{FCS}(\ve{z}) > \tau_n$, and $\Hip_0$ otherwise, where the threshold of the test $\tau_n$ depends on $\alpha$.

In a distributed setting, with noisy links from the CHs to the FC, this statistic provides a lower bound for the false alarm and miss-detection error probabilities. It is worth to mention that these lower bounds are not necessarily tight due to the fact that they do not contemplate the degradation effect introduced by the communication channel between the CHs and the FC. 



\subsection{Local Correlation Strategy}
When the process $\{s_k\}$ is spatially correlated, the covariance matrix $\Xi_{1,n'}$ in (\ref{eq:Xi-np}) is neither diagonal nor block-diagonal. Hence, (\ref{eq:FCS}) cannot be expressed as the sum of the local LLRs from each cluster, and processing the data independently does not lead to the global LLR.
However, it is possible to implement simple distributed detection schemes to build an appropriate statistic at the FC using the compressed data transmitted by the CHs. 

The FCS statistic can be decomposed in such a way that each term is a function of one or at most two sets of measurements from each cluster, i.e., 
\begin{align}
T_\text{FCS}(\ve{z}) 
&= \sum_{i=1}^{m} \ve{z}_i^H \left(\Xi_{0,l_i'}^{-1}-\Xi_{1,l_i'}^{-1} \right)\ve{z}_i \nonumber\\& + \sum_{i=1}^{m} \ve{z}_i^H M_{ii} \ve{z}_i + \sum_{i=1}^{m} \sum_{j=1,j\neq i}^{m} \ve{z}_i^H M_{ij} \ve{z}_j.
\label{eq:FCS-2}
\end{align}
where $\Xi_{0,l_i'}$ and $\Xi_{1,l_i'}$ are the $(i,i)$-th block of $\Xi_{0,n'}$ and $\Xi_{1,n'}$, respectively, $M_{ii}$ is the difference between the $(i,i)$-th block of $\Xi_{1,n'}^{-1}$ and $\Xi_{1,l_i}^{-1}$; and $M_{ij}$ is the $(i,j)$-th block of $\Xi_{1,n'}^{-1}$. The first term in (\ref{eq:FCS-2}) captures most of the correlation between sensor measurements of the same cluster. We refer to this as the \emph{intra-cluster} correlation term. 
The second term captures part of the intra-cluster correlation not included by the first term. The third term considers the correlation between sensor measurements of different clusters and we refer to this as the \emph{inter-cluster} correlation term. 

Considering that the correlation between measurements is a local effect, it is expected that the inter-cluster correlation term does not affect too much the whole statistic if each cluster has an appropriate size respect to the coherence distance of the process\footnote{Measurements taken in sensors separated by a distance greater than the coherence distance of the process are weakly correlated.}. Keeping this effect in mind  and considering also that gathering measurements from distant nodes is energy expensive, it would be reasonable to use a scheme that takes into account only the intra-cluster correlation. Then, the FC could processed the data sent by the CHs without considering the inter-cluster correlation. The strategy is named local correlation strategy (LCS).

Let the local LLR at the $i$-th cluster be:
\begin{align}
T_{\text{LCS},i}(\ve{z}_i) &= \log \frac{p(\ve{z}_i|\Hip_1)}{p(\ve{z}_i|\Hip_0)} 
\equiv \ve{z}_i^H \left( \Xi_{0,l_i'}^{-1}-\Xi_{1,l_i'}^{-1} \right)\ve{z}_i.
\label{eq:Ti:b}
\end{align} 
Each CH is able to exploit the intra-cluster correlation by computing the local LLR $T_{\text{LCS},i}(\ve z_i)$ while keeping the communication energy cost relatively low. Moreover, no additional overhead is necessary to communicate and/or estimate the ACF between measurements of different clusters. Only the local intra-cluster ACF need to be known or estimated at each CH. 
In the next step of the scheme, each CH transmits its corresponding $T_{\text{LCS},i}(\ve z_i)$ to the FC. This communication can be done using a PAC or a MAC. When a PAC is used, since the inter-cluster correlation is discarded, the FC builds the global statistic adding the local statistics provided by the CHs, as suggested by the first term of (\ref{eq:FCS-2}):
\begin{align}
T_\text{LCS}(\ve z)&= \sum_{i=1}^{m} \ve{z}_i^H \left( \Xi_{0,l_i'}^{-1}-\Xi_{1,l_i'}^{-1} \right)\ve{z}_i.\label{eq:LCS-1}
\end{align}
On the other hand, if a MAC is used, the same statistic can be obtained. Notice, however, that the sum of the local statistics is naturally performed by sending them synchronously through the MAC channel. 
The channel communication between the CHs and the FC is assumed to be noiseless. We can justify this saying that usually there are much less CHs than nodes in the network and they typically have some infrastructure that the nodes do not have. For example, the CHs could have high directive antennas pointed to the FC, better energy budgets than the nodes, and, for the strategies considered in this paper, the CHs transmit scalar numbers which could be communicated using a reliable low rate coded digital modulation. For later developments, it is convenient to express the statistic (\ref{eq:LCS-1}) as follows
\begin{align}
T_\text{LCS}(\ve z)&= \ve z^H \Gamma_{n'} \ve z.
\label{eq:LCS-2}
\end{align}
where $\Gamma_{n'} = \diag\left(\Xi_{0,l_1'}^{-1}-\Xi_{1,l_1'}^{-1},\dots,\Xi_{0,l_m'}^{-1}-\Xi_{1,l_m'}^{-1}\right)$ is a $n'\times n'$ block diagonal matrix, with $n'=\sum_{i=1}^{m} l_i'$. 
Notice that, in general, $\Gamma_{n'}\neq \Xi_{0,n'}^{-1} - \Xi_{1,n'}^{-1}$. Hence, the first term of (\ref{eq:LCS-2}) is not exactly the LLR of $\ve z$ although it could be a good approximation if the size of the clusters are greater than the coherence distance of the correlated process.
In the FC, the statistic $T_\text{LCS}(\ve z)$ is compared against a threshold $\tau_n$ in order to decide about the true state of nature. 


Now, we can elaborate an analytical justification about why we ignore the inter-cluster correlation in the LCS strategy using (\ref{eq:FCS-2}). The difference between the FCS and LCS statistics is the matrix used to produce the quadratic form: while $\Xi_{0,n'}^{-1}$ is the same in both statistics because it is a diagonal matrix,  $\Xi_{1,n'}^{-1} \neq \diag (\Xi_{1,l_1'}^{-1},\dots,\Xi_{1,l_m'}^{-1})$. However, it can be shown, by using the Schur complement, that the Taylor approximation of order 0 of the $l_i'\times l_i'$ block matrices of $\Xi_{1,n'}^{-1}$ is exactly the block diagonal matrix $\diag(\Xi_{1,l_1'}^{-1},\dots,\Xi_{1,l_m'}^{-1})$. Clearly, this would be a good approximation when the inter-cluster correlation is weak.

\subsection{Energy detector}
We also consider the statistic usually called energy detector, which does not contemplate the correlation of the sensed process. Therefore, it will result useful to compare it with the previous statistics to show how much one can loose if the correlation is not taken into account. Once each CH has the vector of measurements $\ve z_i$ in (\ref{eq:zCH}), it builds $T_{\text{ED},i}(\ve z_i) = \|\ve z_i\|^2$ and sends it through a PAC or a MAC to the FC.
The energy detector available in the FC is defined by
\begin{align}
T_\text{ED}(\ve{z}) = \|\ve z\|^2.
\label{eq:TED}
\end{align}
Notice that this statistic is optimal only when the process $\{s_k\}$ is white noise process.


\subsection{Clustering partition problem}
\label{:clustering}
Now, we are ready to formulate the main problem of this paper. We first summarize the whole strategy used. The $n$ sensors of the WSN take measurements and transmit them to their corresponding CHs using the PFS-MAC strategy or the AFS-PAC strategy. Each CH has access to (\ref{eq:zCH}) and allows the FC to build a statistic $T_j$ using one of the previous strategies: $j\in\{\text{FCS, LCS, ED}\}$. The FC makes a decision about the state of the nature comparing $T_j$ against threshold  $\tau_j$. The false alarm and the miss-detection probability are, respectively, $P_\text{fa}^n=\prob_0(T_j>\tau_j)$ and $P_\text{m}^n=\prob_1(T_j<\tau_j)$.
 
The way each sensor node is assigned to each cluster is an important issue affecting the performance of the network. One may wonder which is the best allocation possible and which is its level of performance. To that end we can fix the probability of false alarm to a level $\alpha\in (0,1)$ and consider an energy constraint given by (\ref{eq:SensorEC}). Let $e_k$ be the label of the $k$-th sensor that identifies to which cluster it belongs, i.e.,  $e_k = i$ if $k$ belongs to the $i$-th cluster. We can then formulate a problem in order to obtain the optimal clustering partition that could be cast as follows:
\begin{gather}
 \inf_{(e_1,\dots,e_n)\in[1:m]^n} P_\text{m}^n,\nonumber\\ 
  \mbox{s.t.}\ P_\text{fa}^n\leq \alpha,
  E_k \leq \bar{E},\ \forall k\in[1:n].
  \label{eq:optimum}
\end{gather}
where the threshold of the test $\tau_j$ depends on $\alpha$. 
This is a challenging non-linear integer programing problem with constraints. In fact, it is a NP-hard problem and its solution is out of the scope of this paper. As our objective is to understand the trade-off between the beneficial use of spatial correlation against energy consumption in order to overcome the path-loss effect, we will consider a simpler clustering partition problem. We will assume that all the clusters have the same shape, in particular, we will assume that each of the $m$ clusters are square regions of area $l_i A_0$,  where $l_i$ is the number of sensors in the $i$-th cluster, and  $nA_0=\sum_{i=1}^m l_iA_0$ is total area in which the network is deployed. 
This simple model captures the basic fact that when the number of clusters is large, the size of them will be small, and the path loss effect in the local transmissions will be less harmful. However, as the clusters are smaller, less sensor nodes are present within them. In addition, some other nodes in their local neighborhoods, which can have strongly correlated measurements, could be assigned to other adjacent clusters. As the clusters do not cooperate and they make their processing independently, this clearly reduces the benefit of using the spatial correlation in each cluster. 
This basic model will be used in the following section with minor variations. For example, in Section \ref{sec:equicorr} we will assume that each cluster has exactly $l$ nodes arbitrary distributed in space. On the other hand, in Section \ref{sec:randomNetwork}, we will assume that the nodes are spatially distributed as a homogeneous Poisson Point Process (PPP) \cite{stochastic_geometry2009} where the average number of sensor per cluster is $l$, although for each realization of the spatial random process each cluster could have different number of nodes. In both cases, we variate the area of each cluster and, consequently, the number of sensors per cluster $l$ (or its average) but keeping constant the total amount of nodes $n$ in the sensor network (or its average). 



\section{Computation of the error probabilities}
\label{sec:errorProb}
In this section, we present a result used to compute tight approximations of the miss-detection probability $P_\text{m}^n$ and the false alarm probability $P_\text{fa}^n$ using the LMGFs of the statistic under $\Hip_0$ and $\Hip_1$.

\begin{theorem}[]
\label{thm:Pe}
Let $\{y_k\}_{k=1}^n$ be mutually independent random variables with probability density function (PDF) $p_k$ and LMGF $\mu_{k}(s)\defeq \log\Ex\left(e^{y_k s}\right)$. Assume that $\Ex(y_k^2)$ and $\Ex(|y_k-\Ex(y_k)|^3)$ exist and are finite  $\forall k\in[1:n]$. Let $T_n=y_1+\dots + y_n$ with LMGF $\mu_{T_n}\defeq \log\Ex\left(e^{T_n s}\right)$ and let $\tau_n\in \R$. Then, if $\tau_n> \Ex(T_n)$,  
\begin{align}
&\prob(T_n>\tau_n) =\nonumber\\
&\left(\tfrac{1}{\sqrt{2\pi s_0^2\ddot{\mu}_{T_n}(s_0)}} + \mathcal{O}\left(\tfrac{1}{\sqrt{n}}\right)\right) e^{-(s_0\dot{\mu}_{T_n}(s_0) - \mu_{T_n}(s_0))} \label{eq:generalPfa}
\end{align}
where $s_0>0$ is the solution to $\tau_n = \dot{\mu}_{T_n}(s_0)$. On the other hand, if $\tau_n < \Ex(T_n)$,  
\begin{align}
&\prob(T_n < \tau_n) =\nonumber\\
&\left(\tfrac{1}{\sqrt{2\pi s_1^2\ddot{\mu}_{T_n}(s_1)}} + \mathcal{O}\left(\tfrac{1}{\sqrt{n}}\right)\right) e^{-(s_1\dot{\mu}_{T_n}(s_1) - \mu_{T_n}(s_1))} \label{eq:generalPm}
\end{align}
where $s_1<0$ is the solution to $\tau_n = \dot{\mu}_{T_n}(s_1)$.
\end{theorem}
\begin{proof}
See App. \ref{app:EPA-A}.
\end{proof}

\begin{corollary}[Gaussian quadratic form]
\label{cor:GQF}
Consider the statistic given by the quadratic form 
\begin{equation}
T_{n'}=\ve{z}^H P_{n'} \ve{z},
\label{eq:GaussQF}
\end{equation}
where the distribution of $\ve{z}$ under $\Hip_i$ is $\CN(\ve{0},\Xi_{i,n'})$, $i=0,1$. Assume that $\Xi_{i,n'}$ and $P_{n'}$ are positive definite matrices and that $\tau_{n'}\in (\Ex_0(T_{n'}),\Ex_1(T_{n'}))$ is the threshold of the test. The false alarm and the miss error probabilities are, respectively, 
\begin{align}
P_\text{fa}^{n'} &= \left(\tfrac{1}{\sqrt{2\pi s_0^2 \tr\{[(\Xi_{0,n'}P_{n'})^{-1}-s_0 I_{n'}]^{-2}\}}} + \mathcal{O}\left(\tfrac{1}{\sqrt{{n'}}}\right)\right) \times\nonumber\\ 
&\times e^{- s_0 \,\tr\{[(\Xi_{0,n'}P_{n'})^{-1}-s_0 I_{n'}]^{-1}\} - \log\det\left(I_{n'} - s_0 \Xi_{0,{n'}}P_{n'} \right)}, \nonumber\\
P_\text{m}^{n'} &= \left(\tfrac{1}{\sqrt{2\pi s_1^2 \tr\{[(\Xi_{1,n'}P_{n'})^{-1}-s_1 I_{n'}]^{-2}\}}} + \mathcal{O}\left(\tfrac{1}{\sqrt{n'}}\right)\right) \times\nonumber\\ 
&\times e^{- s_1 \,\tr\{[(\Xi_{1,n'}P_{n'})^{-1}-s_1 I_{n'}]^{-1}\} - \log\det\left(I_{n'} - s_1 \Xi_{1,n'}P_{n'} \right)}, \nonumber
\end{align}
where $s_0>0$ and $s_1<0$ are the solution to $\tau_{n'} = \tr\{[(\Xi_{0,n'}P_{n'})^{-1}-s_0 I_{n'}]^{-1}\}$ and $\tau_{n'} = \tr\{[(\Xi_{1,n'}P_{n'})^{-1}-s_1 I_{n'}]^{-1}\}$, respectively, and make $I_{n'} - s_i \Xi_{i,{n'}}P_{n'}$ positive definite for $i=0,1$.
\end{corollary}

\begin{proof}
See App. \ref{app:EPA-B}.
\end{proof}
\section{Performance Analysis}
\label{sec:performance}

In Section \ref{sec:equicorr}, we consider a case of study where the process $\{s_k\}$ is assumed to be equicorrelated, a particular case of a stationary process. It properly models the spatial correlation between the nodes and allows to obtain closed forms expressions of the error probabilities. 
In Section \ref{sec:randomNetwork}, we will consider that the nodes are spatially distributed as a PPP.  
We will see that similar conclusions to those obtained for the equicorrelated process can be drawn for non-stationary process.
Before continuing, let $\snr_\mathrm{M} = \frac{\sigma^2_s}{\sigma^2_v}$ and $\snr_\mathrm{C} = \frac{\bar{E}}{\sigma^2_w}$ be the signal to noise ratio of the measurements and of the communication channel, respectively.

\subsection{A case of study: the equicorrelated process}
\label{sec:equicorr}
Consider a stationary zero-mean circular-symmetric Gaussian equicorrelated process with  ACF $R_s(k) = \sigma_s^2(1-\rho) \delta_{k0} + \sigma_s^2\rho$, where $\delta_{k0}$ is the delta of Kronecker. 
The PSD of the equicorrelated process is $\phi(\nu)= \sigma_s^2(1-\rho) + \sigma_s^2\rho \delta(\nu)$, where $\delta(\nu)$ is the Dirac delta distribution and $\rho$ is a scalar parameter that controls the correlation, $0 \leq\rho<1$.
The covariance matrix for the measurements of each cluster is given by $\Sigma_{s,l}=\sigma_s^2(1-\rho) I_l + \sigma_s^2\rho \ve 1_l$ where $\ve 1_l$ is a $l\times l$ matrix of ones. 

As the equicorrelated process is stationary, both PFS-MAC or AFS-PAC strategies can be used. However, in this section we will focus on the first one. 
We assume that all the clusters are identical and have the same parameters in Def \ref{def:PFS-MAC}: $l_i=l$, $l_{i}'=l'$, $\beta_i=\beta$, $d_i=d_\text{max}(l)$ and $\gamma_i=\gamma_0$ $\forall i$.
Considering that the CH could be in any position inside the coverage area $l A_0$ of its corresponding cluster, and that each cluster is a square region, the maximum possible distance between a node and its corresponding CH is $d_\text{max}(l) =\sqrt{2 l A_0}$.

Let $\snr_\text{PFS-MAC}(l)=\frac{{\gamma_0}^2\sigma^2_s}{{\gamma_0}^2\sigma^2_v+\sigma_w^2}$ be the effective signal-to-noise ratio of the measurements available in each CH. If we replace (\ref{eq:Prec-PFS-MAC}) in the previous definition, we have
\begin{equation}
\snr_\text{PFS-MAC}(l) = \frac{\snr_\mathrm{M}\snr_\mathrm{C}}{\snr_\mathrm{C} + (1+\snr_\mathrm{M}) (2 l A_0)^\frac{\epsilon}{2}},
\label{eq:SNR-PFS-MAC}
\end{equation}
where we have emphasized that the effective signal-to-noise ratio is a function of $l$. 
 
In order to evaluate the error probabilities of the LCS, FCS, and ED schemes, we need to compute the LMGFs of each statistic under $\Hip_i$ and their first and second derivatives to then apply Th. \ref{thm:Pe}. For this particular process, it is possible to obtain closed forms solutions for the error probabilities for any threshold $\tau_n$. 
However, these expressions are lengthy and are not shown here. The symbolic expressions are shown in the MATLAB script \emph{SymExprForPfaPm.m} available with this paper.
 
Therefore, we make focus on the case when the threshold for each strategy is set such that the miss error exponent is zero, where simple expressions are obtained. 
We summarize the results in the following lemma.

\begin{lemma}[Probabilities of false alarm for the equicorrelated process using the PFS-MAC strategy]
\label{lem:FAP}
Consider the following thresholds such that the miss error exponents are zero for the three strategies: $\tau_\text{LCS}=\Ex_1(T_\text{LCS}(\ve{z}))$, $\tau_\text{FCS}=\Ex_1(T_\text{FCS}(\ve{z}))$ and $\tau_\text{ED}=\Ex_1(T_\text{ED}(\ve{z}))$. Then, the false alarm probabilities of the LCS, FCS and ED strategies using PFS-MAC are, respectively,  
\begin{align}
& P^{n}_\mathrm{fa,LCS}= \left(\tfrac{1}{\Gamma \sqrt{2\pi n\left[ (l+\beta -2)\rho^2 +2(1-\beta)\rho +\beta\right]} } + \mathcal{O}\left(\tfrac{1}{\sqrt{n}}\right)\right)\times \nonumber\\
& e^{ -n(\Gamma(\rho+\beta(1-\rho)) -\frac{1}{l}\log(1+ \Gamma(1+(l-1)\rho))  -(\beta-\frac{1}{l})\log(1+ \Gamma(1-\rho)))}
\label{eq:PfaEqui-LCS}\\
& P^{n}_\mathrm{fa,FCS}= \left(\tfrac{1}{\Gamma \sqrt{2\pi n\left[ (n+\beta -2)\rho^2 +2(1-\beta)\Gamma +\beta\right]} } + \mathcal{O}\left(\tfrac{1}{\sqrt{n}}\right)\right)\times \nonumber\\
& e^{ -n(\Gamma(\rho+\beta(1-\rho)) -\frac{1}{n}\log(1+ \Gamma(1+(n-1)\rho)) -(\beta-\frac{1}{n})\log(1+ \Gamma(1-\rho)))}
\label{eq:PfaEqui-FCS}\\
& P^{n}_\mathrm{fa,ED}= \left(\tfrac{1}{\Gamma \sqrt{2\pi n\beta}} + \mathcal{O}\left(\tfrac{1}{\sqrt{n}}\right)\right) e^{-n\beta(\Gamma -\log(1+ \Gamma))},
\label{eq:PfaEqui-ED}
\end{align}
where $\Gamma\defeq\snr_\text{PFS-MAC}(l)$.
\end{lemma}
\begin{proof}
See App. \ref{app:FACS}
\end{proof}

It is easy to show that $P^{n}_\mathrm{fa,LCS}$ decreases monotonically with $n$, $\beta$ and $\rho$. If $\snr_\text{PFS-MAC}(l)$ were independent of $l$, $P^n_\mathrm{fa,LCS}$  would also decrease monotonically with $l$ and $\snr_\text{PFS-MAC}(l)$. However, as it is seen in (\ref{eq:SNR-PFS-MAC}), $\snr_\text{PFS-MAC}(l)$ decreases with $l$, and therefore, a trade-off between these two parameters is established: for a given spatial density of sensors, to increase $l$, the number of nodes on each cluster, permits to successfully exploit the correlation, although this produces a worse $\snr_\text{PFS-MAC}(l)$ at each cluster head. If the cluster size is reduced, $\snr_\text{PFS-MAC}(l)$ improves but the correlation is not fully exploited. 




In Fig. \ref{fig:PfaVl}, the false alarm probability of the three schemes LCS, FCS and ED given by (\ref{eq:PfaEqui-LCS}), (\ref{eq:PfaEqui-FCS}) and (\ref{eq:PfaEqui-ED}) respectively, are plotted against $l$, for the parameters shown in the caption. 
\begin{figure}[tbh]
\centering
\includegraphics[width=1\linewidth]{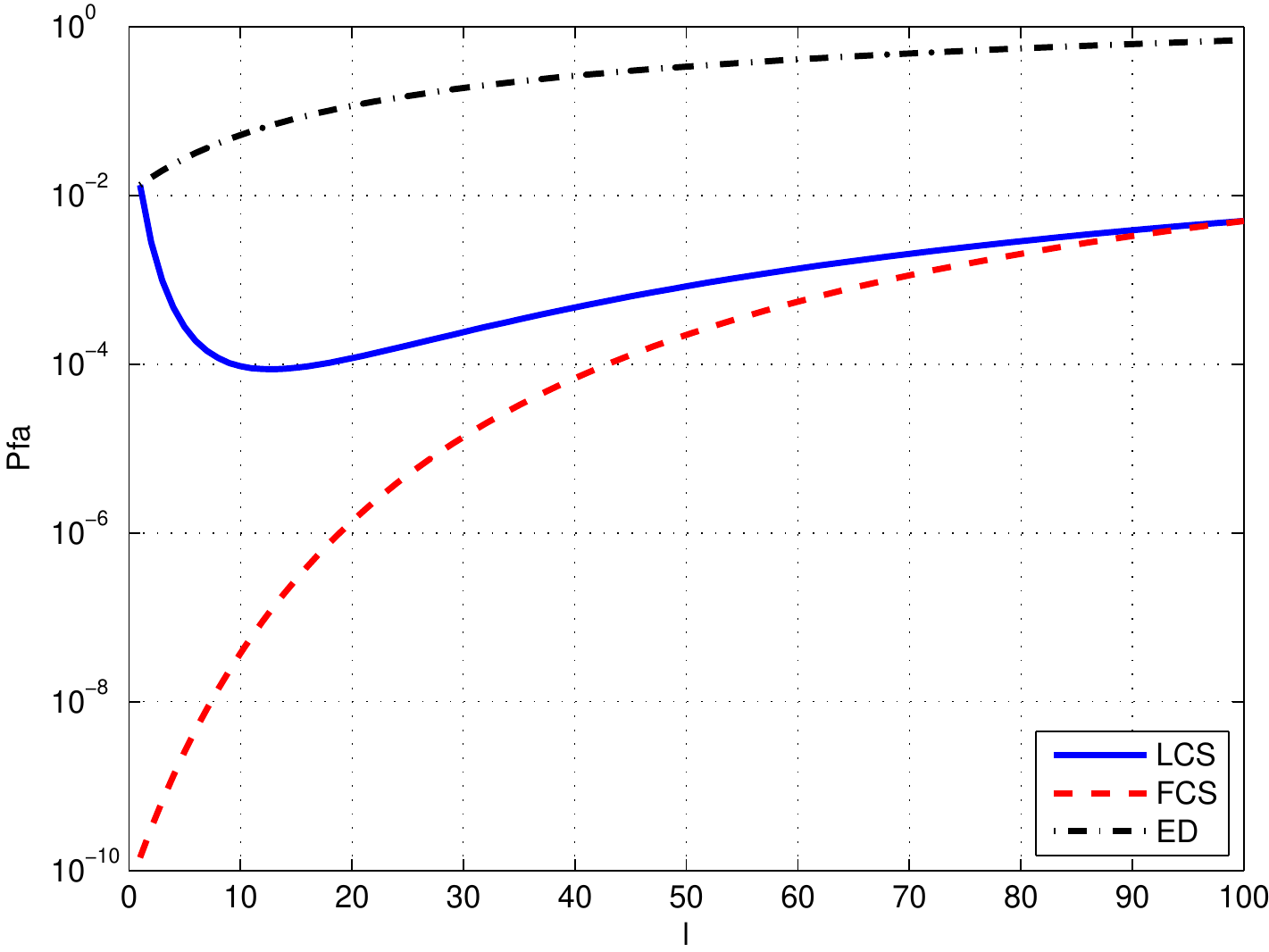}
\caption{Trade-off between exploiting correlation and improving the measurement SNR. Parameters: $\snr_\text{C}=12$ dB, $\snr_\text{M}=-6$ dB, $\rho=0.9$, $\epsilon = 2$, $\beta =1$ and $n=100$.}
\label{fig:PfaVl}
\end{figure}
We observe that for the LCS scheme, to increase the size of the cluster allows to exploit the correlation of the process and decrease $P_\text{fa,LCS}^n$. Beyond a certain cluster size, say $l=10$, the benefit of exploiting the correlation is overcome by the lost of $\snr_\text{PFS-MAC}(l)$, therefore, causing a worse performance. In the case of ED, increasing the cluster size reduces the effective $\snr_\text{PFS-MAC}(l)$ without exploiting the correlation, due to the nature of this scheme. On the other hand, as the FCS scheme has access to the whole vector of measurements $\ve z$, and computes the optimum LLR detector, it is only limited by $\snr_\text{PFS-MAC}(l)$, which decreases when $l$ increases. For a fixed cluster size, the FCS performance allows to quantify how much  one can gain potentially by exploiting the correlation of the sensed process. For example, if $l=10$, one could expect to decrease $P_\text{fa,LCS}^n$ between 3 and 4 orders of magnitude, although this is not entirely true given that the FCS scheme permits the FC to have access to the full vector $\ve z$ neither being penalized in bandwidth nor in transmitted energy, i.e., it could be a loose lower bound for some $l$'s. Additionally, in the case of the PFS-MAC scheme, the energy received in the CHs is limited by the furthest node, which could impose severe energy limitations. 
Notice that by definition of the LCS strategy, it coincides with the ED strategy when $l=1$ (no correlation is considered in the statistic) and with FCS strategy when $l=n$ (the network has a unique cluster).

In Fig. \ref{fig:PfaVsnrC}, we plot the false alarm probability of the three statistics in Lem. \ref{lem:FAP} against $\snr_\text{C}$ with $l$ as a parameter. In Fig. \ref{fig:PfaVsnrC_LCS} we plot only the LCS strategy for the sake of clarity. Two regimes can be observed there: i) an energy-limited regime for which $\snr_\text{C}$ limits the performance and, ii) a correlation-limited regime, where the performance is limited due to fact that the inter-cluster correlation is not considered. When the cluster size is small, e.g. low number of sensors per cluster $l$, the sensors are near to the CH and the communication between them is reliable. We see in Fig. \ref{fig:PfaVsnrC_LCS} that a WSN with small clusters performs better when the sensor energy budget is tight or the path loss effect is severe.  However, when the sensor energy budget and the path loss effect are not an issue, large clusters are preferred because they can exploit successfully the intra-cluster correlation, obtaining large correlation gains even for weakly correlated processes ($\rho = 0.2$), as indicated in the figure. 
\begin{figure}[bth]
\centering
\subfigure[LCS strategy. \label{fig:PfaVsnrC_LCS}]{\includegraphics[width=1\linewidth] {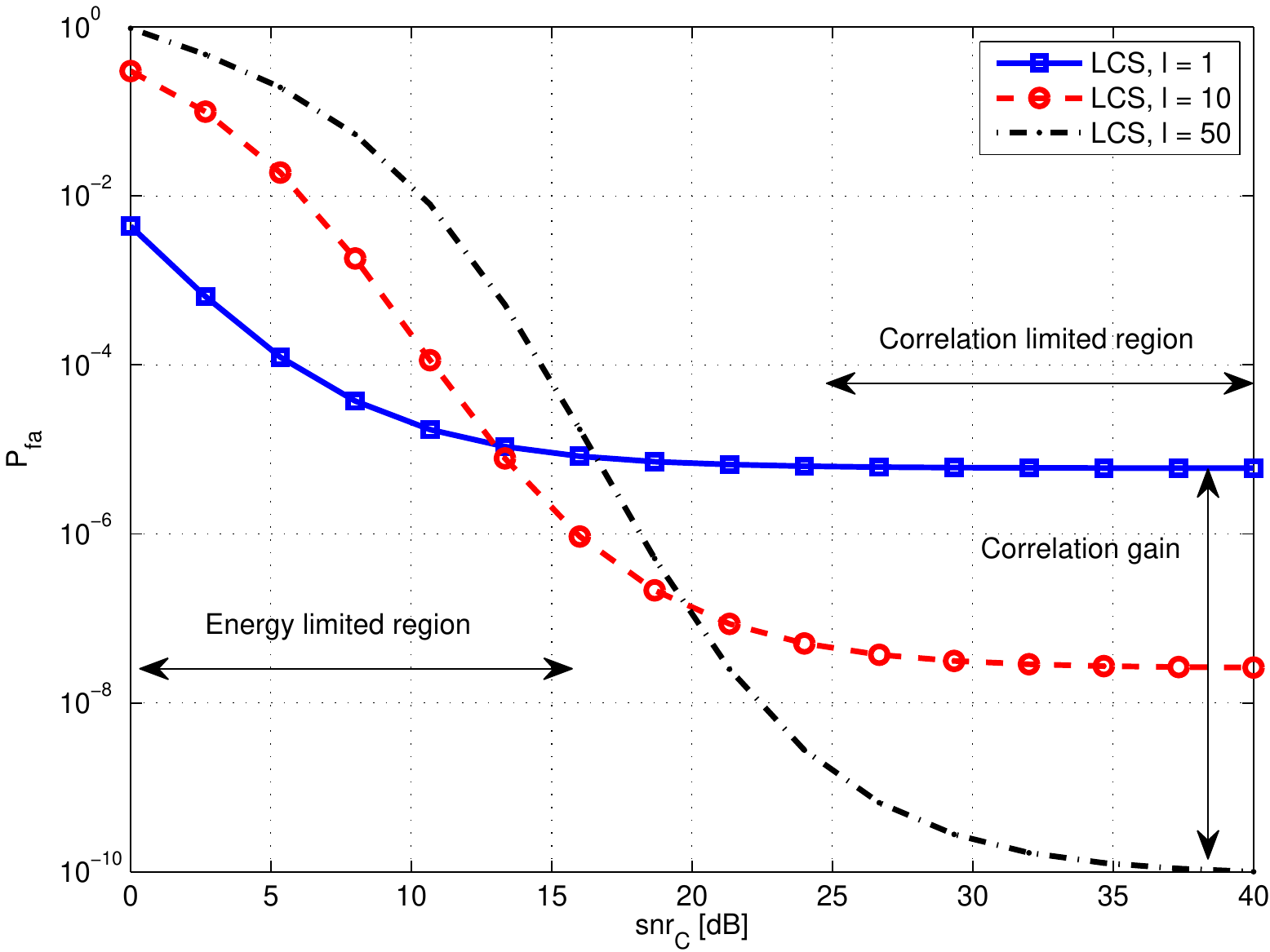}}
\subfigure[FCS and ED strategies. \label{fig:PfaVsnrC_FCS_ED}]{\includegraphics[width=1\linewidth] {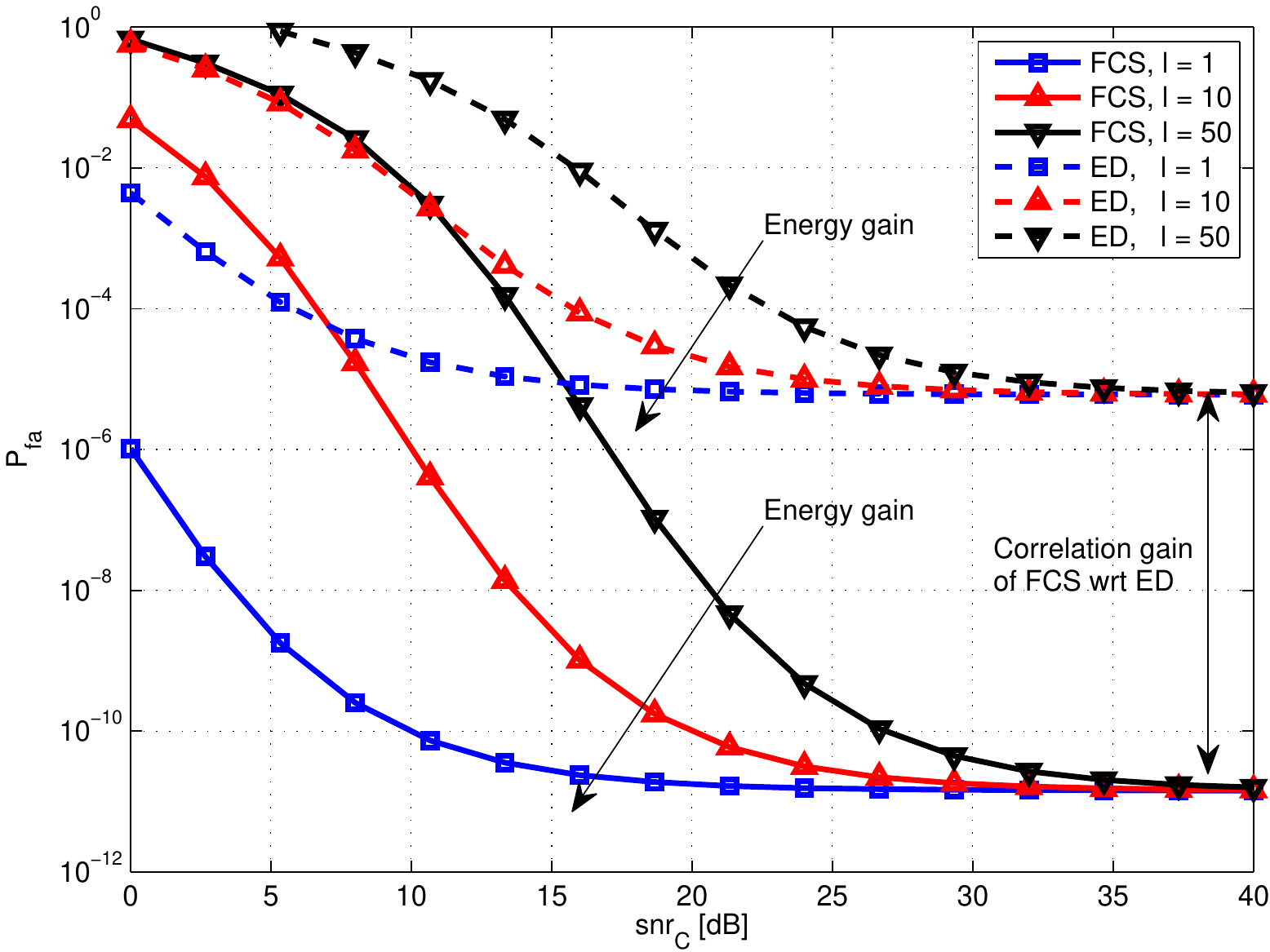}}
\caption{Probability of false alarm versus communication signal to noise ratio with parameters $\snr_\text{M}=-3$ dB, $\rho=0.2$, $\epsilon =2$, $\beta = 1$ and $n=100$.}
\label{fig:PfaVsnrC}
\end{figure}
In Fig. \ref{fig:PfaVsnrC_FCS_ED}, we show the behavior of the FCS and the ED strategies. We see that both strategies have lesser false alarm probabilities for small cluster sizes for the entire $\snr_\text{C}$ range. In other words, their performance is better due to an energy gain of the measurements in the CHs. Neither FCS nor ED modify how they consider the correlation of the process when the cluster size changes: the ED strategy always obviates the correlation while the FCS always considers the full correlation of the process (intra and inter-cluster). On the other hand, the proposed LCS makes the compromise between reliable measurement transmission to the CHs and exploiting the correlation of the process. For a given $l$, the FCS and the ED strategies provide a lower bound and an upper bound, respectively, of the error probability of the LCS strategy.


Similarly to Lem. \ref{lem:FAP}, we use Th. \ref{thm:Pe} to compute the miss-detection probability for the three statistics LCS, FCS and ED, when the threshold for each test $\tau_n$ is such that the false alarm error exponent is zero, i.e., $\tau_n$ is the mean of each statistic under $\Hip_1$. The expressions are lengthy and are not shown here. We have plotted them in Fig. \ref{fig:PmVsnrC_LB}, against $\snr_\text{C}$ with $l$ as a parameter.
It is observed that, similarly to the probability of false alarm, two regimes are observed for the miss-detection probability: a energy-limited regime and an correlation-limited regime. For a low $\snr_\text{C}$, the miss-detection probability is limited by $\snr_\text{PFS-MAC}(l)$. Therefore, a lower $P_\text{m}$ is obtained when the cluster size is small (low $l$). For a high $\snr_\text{C}$, the performance is limited by how much the correlation is exploited. As $\snr_\text{PFS-MAC}(l)$ is not an issue, the LCS can increase the cluster size in order to exploit the intra-cluster correlation and decrease $P_\text{m}$.  
\begin{figure}[tb]
\centering
\includegraphics[width=1\linewidth]{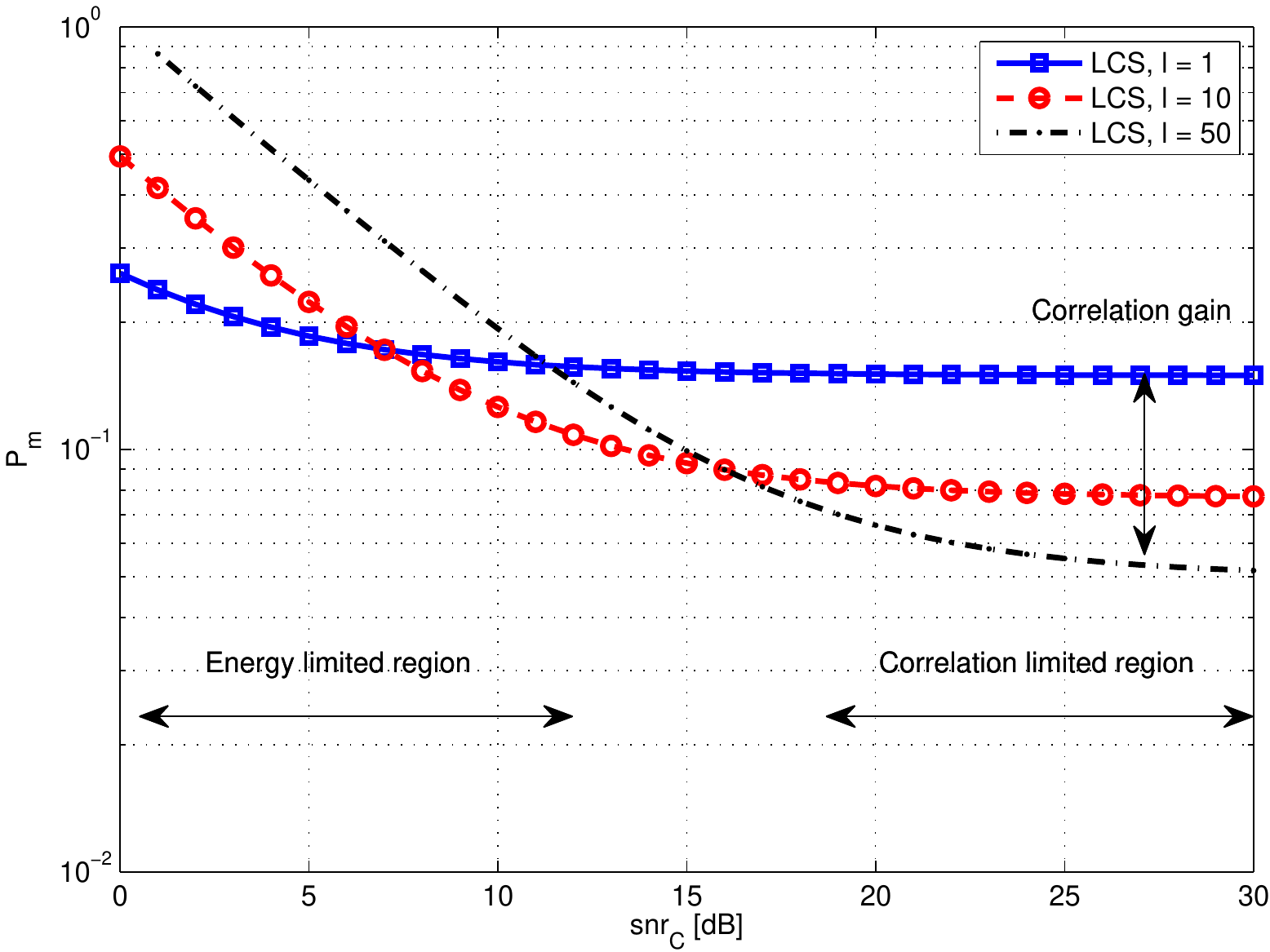}
\caption{$P_\text{m}$ versus $\snr_\text{C}$. $\snr_\text{M}=0$ dB, $n=100$ and $\rho=0.9$.}
\label{fig:PmVsnrC_LB}
\end{figure}

To conclude this section, we obtain the optimum size $l_\text{opt}$ which minimizes the false-alarm error probability of the LCS strategy given by (\ref{eq:PfaEqui-LCS}) among all cluster sizes $l\in [1:n]$ such that $n=ml$. This is done numerically given that there is not a closed-form expression for $l_\text{opt}$. In Fig. \ref{fig:loptVsnrC_psnrM}, we plot $l_\text{opt}$ versus $\snr_\text{C}$ with $\snr_\text{M}$ as a parameter. We see that the optimum cluster size increases with $\snr_\text{C}$ for each $\snr_\text{M}$. We also observe that for a fixed $\snr_\text{C}$, $l_\text{opt}$ is greater for lower $\snr_\text{M}$. That is, as $\snr_\text{M}$ decreases, it is better to exploit the intra-cluster correlation increasing $l$ than to increase $\snr_\text{PFS-MAC}(l)$. This is an important guideline for the case in which the WSN is composed by cheap nodes with limited sensing capabilities. In Fig. \ref{fig:loptVsnrC_pCorr}, we plot $l_\text{opt}$ versus $\snr_\text{C}$ with $\rho$ as a parameter. For a given $\snr_\text{C}$, we see that it is more profitable to increase the size of the cluster when the process is highly correlated than to increase $\snr_\text{PFS-MAC}(l)$.
\begin{figure}[htb]
\centering
\subfigure[$\snr_\text{M}$ as a parameter. \label{fig:loptVsnrC_psnrM}]{\includegraphics[width=1\linewidth]{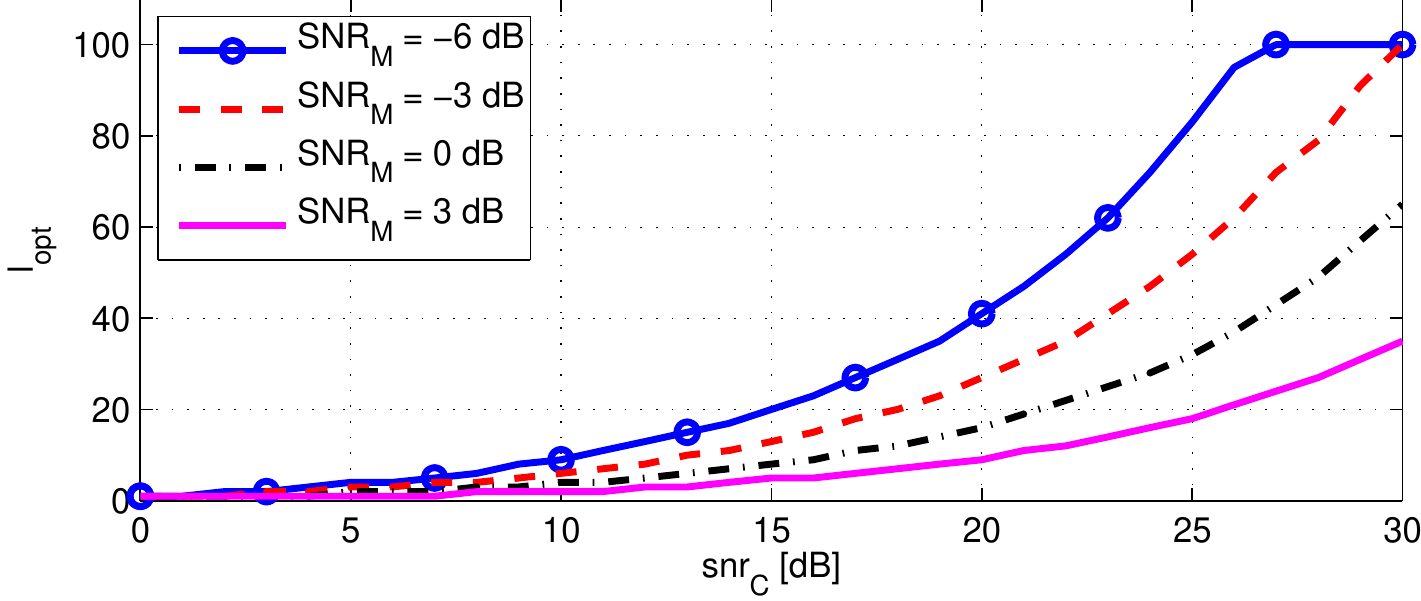}}
\subfigure[$\rho$ as a parameter. \label{fig:loptVsnrC_pCorr}] {\includegraphics[width=1\linewidth]{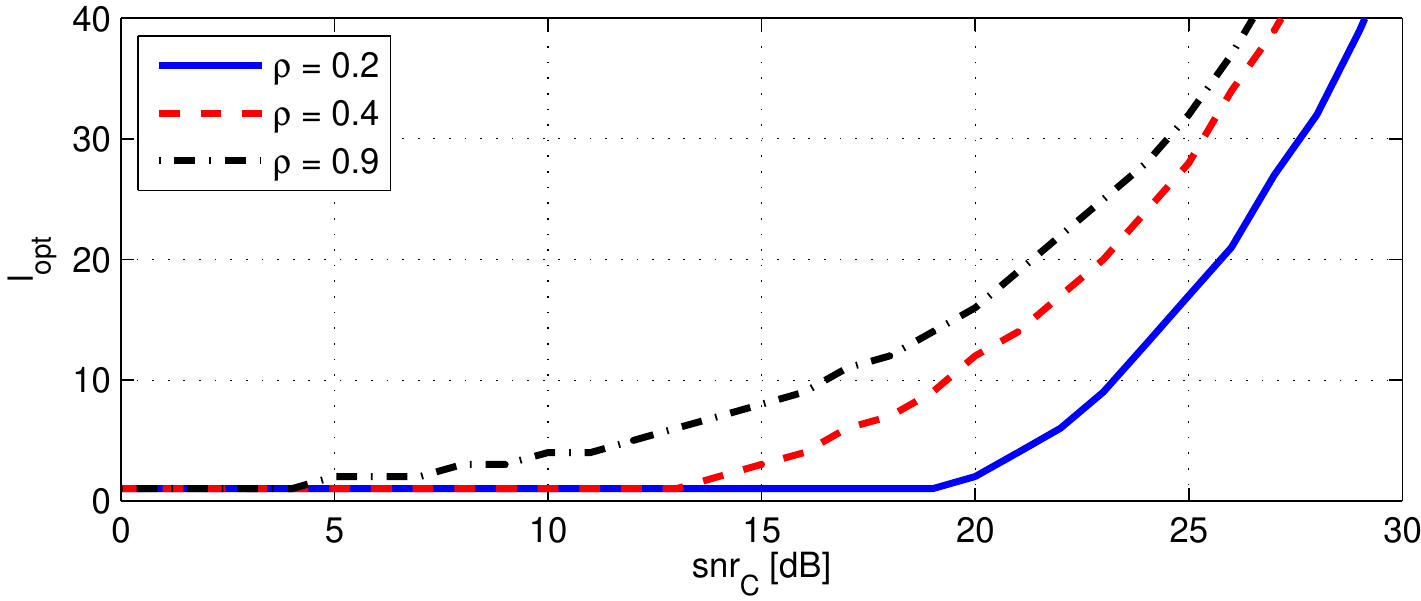}}
\caption{$l_\text{opt}$ versus $\snr_\text{C}$ with $\epsilon = 2$, $\beta=1$ and $n=100$.}
\label{fig:loptVsnrC}
\end{figure}

The optimum cluster size $l_\text{opt}$ allows to redefine the idea of spatial coherence regions (SCRs). Typically, a SCR is an inherent property of the process to be sensed, uniquely determined by the spatial correlation between its samples \cite{Sayeed2003DistClassifClusters}. In this work, guided by the error probability as an indicator of the ultimate performance, we obtain, that a SCR is not only defined by the spatial correlation of the process, but also by critical parameters as $\snr_\text{M}$, transmitted energy $\bar{E}$, path loss exponent $\epsilon$, and bandwidth $\beta$ available at each cluster.

\subsection{Random Network}
\label{sec:randomNetwork}
In this section, we compute the performance of the statistics when the sensors of the network  are spatially distributed as a  homogeneous Poisson Point Process \cite{stochastic_geometry2009} of intensity $\lambda>0$. In this case, the total number of sensors and the number of sensors per cluster follow a Poisson distribution with mean $\lambda n A_0$ and $\lambda l A_0$, respectively. For simplicity, we continue using the assumption that the clusters are squares of area $l A_0$ and the region to be sensed has a total area $n A_0$. We set $\lambda=1$ and $A_0=1$. Therefore, the average number of sensors per cluster coincides with $l$. We also assume that the spatial process is isotropic meaning that the ACF depends only on the distance $d$ between the points for which it is evaluated. Some examples of possible isotropic ACFs are the equicorrelated process presented previously, the exponential ACF $R_s(d)=\sigma_s^2 \rho^{d}$ or the hyperbolic ACF $R_s(d)=\sigma_s^2 \frac{\rho}{\rho + d}$. The exponential ACF was found to be a good model for the shadowing effect in a urban cellular environment \cite{Gudmundson1991} and it will be used here. 

The precoding strategy used for an arbitrary spatial distribution of the nodes is AFS-PAC whose gains are defined in (\ref{eq:Prec-AFS-PAC}). For this strategy, the signal-to-noise ratio for the $k$-th node at the $i$-th CH is defined as
$\snr_\text{AFS-PAC}(k)= \frac{c_{kk}^2 \sigma_s^2}{c_{kk}^2 \sigma_v^2 + \sigma_w^2/a_k^2}$. Replacing (\ref{eq:Prec-AFS-PAC}) in it gives
\begin{equation*}
\snr_\text{AFS-PAC}(k)= \frac{\snr_\text{C}\snr_\text{M}}{\snr_\text{C} + (1+\snr_\text{M})\|\ve x_k - \ve x_{\text{CH},i}\|^\epsilon}.
\end{equation*}
We compute the average miss error probability $\bar{P}_\text{m}$ as a metric of performance, defined by the average over the PPP $\Phi$, $\bar{P}_\text{m} = \Ex(P_\text{m}(\Phi))$ subject to $P_\text{fa}(\Phi)=\alpha$ for each realization of the PPP. This average is computed using the method of Monte Carlo and the theoretical expressions found in Corollary \ref{cor:GQF} (labeled LCS, FCS and ED in Figs. \ref{fig:PmVsl_Rnd} and \ref{fig:PmVsSNRc_Rnd}). On the other hand, we also generate the sensed process and the corresponding noises using again the method of Monte Carlo to validate the theoretical expressions derived in Corollary \ref{cor:GQF} (labeled LCS-MC, FCS-MC and ED-MC in Figs. \ref{fig:PmVsl_Rnd} and \ref{fig:PmVsSNRc_Rnd}). 

In Fig. \ref{fig:PmVsl_Rnd} we plot the miss probability as a function of the average number of nodes per cluster $l$ with the parameters indicated in its caption. Notice that the results for a random network with exponential ACF are similar to the equicorrelated process case. Considering the LCS strategy, in spite of decreasing the effective signal-to-noise ratio at each CH $\snr_\text{AFS-PAC}(k)$ for most of the sensors when $l$ increases up to 10, the miss error probability decreases because of the intra-cluster correlation is better exploited. If the cluster size increases further, the path loss effect inside each cluster dominates and the performance is deteriorated.
The LCS strategy is close to the bound provided by the FCS scheme and the gain compared with the ED is considerable.
\begin{figure}[bth]
\centering
\includegraphics[width=1\linewidth]{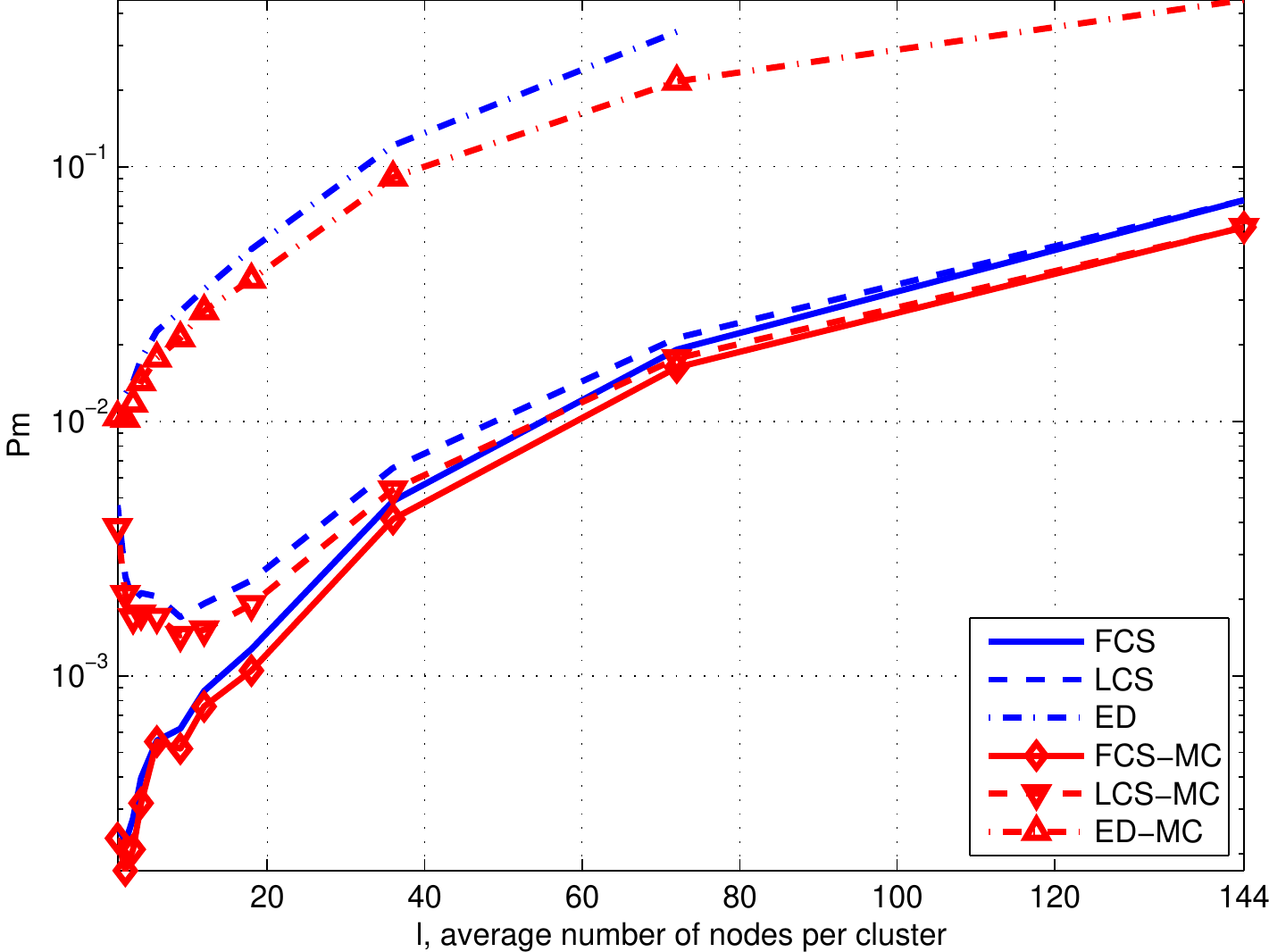}
\caption{$\bar{P}_\text{m}$ versus $l$ for an exponential ACF with $\rho=10$, $n= 144$, $P_\text{fa}= 10^{-2}$, $\snr_\text{M}=0$ dB and $\snr_\text{C}= 12$ dB.}\label{fig:PmVsl_Rnd}
\end{figure}

In Fig. \ref{fig:PmVsSNRc_Rnd} we plot the miss error probability as a function of $\snr_\text{C}$ with the parameters indicated in its caption. Notice again that the results are similar to the closed-form expressions obtained for the equicorrelated process. The LCS strategy is close to the bound provided by the FCS scheme and the gain respect to the ED strategy is considerable for the entire shown range of $\snr_\text{C}$. These Monte Carlo curves also validate the results of Th. \ref{thm:Pe}, where the central limit theorem was applied to compute the error probabilities, obtaining good approximations even for a low number of sensors $n$.   
\begin{figure}[htb]
\centering
\includegraphics[width=1\linewidth]{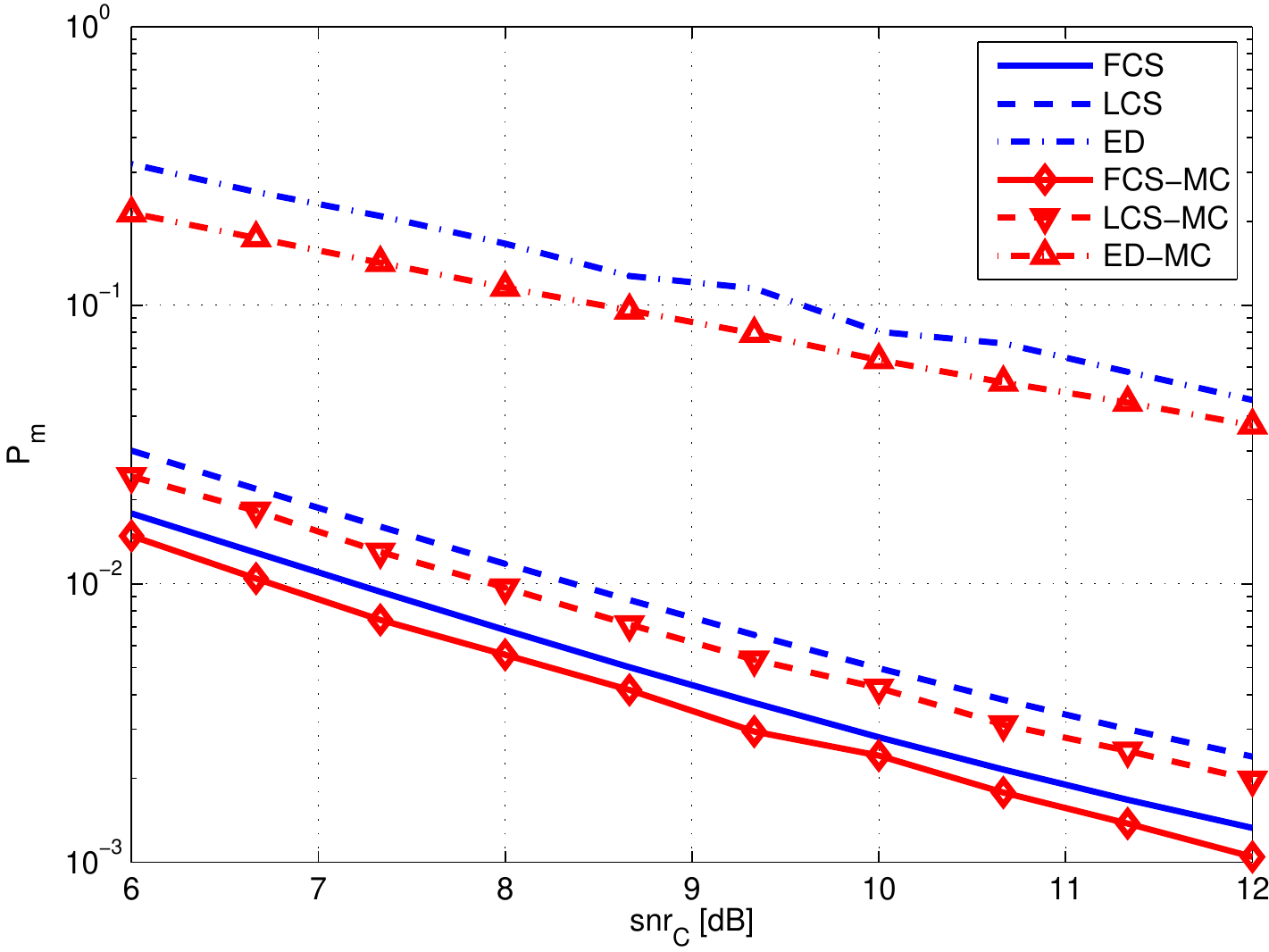}
\caption{$P_\text{m}^n$ vs $\snr_\text{C}$ for an exponential ACF with $\rho=10$, $n= 144$, $l=18$, $m=8$, $P_\text{fa}= 10^{-2}$, and $\snr_\text{M}=0$ dB.}
\label{fig:PmVsSNRc_Rnd}
\end{figure}

\section{Conclusions}
\label{sec:conclusions}
In this paper, we developed a model and appropriate tools for evaluating the performance of a distributed detection scheme that makes a trade-off between transmission energy and exploitation of the spatial correlation of the process. We characterized this trade-off and found two regimes of operation: an energy limited regime, and a correlation-limited regime. We showed that the proposed LCS strategy allows to obtain excellent performances, which are close to the ones obtained with the use the global network-wide correlation information. The use of the spatial correlation is of great importance in  WSNs where the captured measurements  have low signal-to-noise ratio and also when the process to be sensed presents moderate to high correlation.

\appendices

\section{Error Probability Approximations}
In this appendix we derive an approximation for both the false alarm probability $P_\text{fa}^n$ and the miss-detection probability $P_\text{m}^n$ following a line reasoning similar to \cite{Gallager1968IT} and using a version of the central limit theorem (CLT) for a sum of independent but not identically distributed random variables given by the following lemma. 
\begin{lemma}[Essen CLT]
\label{lem:CLT}
Let $x_k$ be mutually independent random variables such that 
\begin{equation*}
\Ex(x_k)=0, \qquad \Ex(x_k^2)=\sigma_k^2, \qquad \Ex(|x_k|^3)=\rho_k.
\end{equation*}
Let $s_n^2 = \sigma^2_1 + \dots +\sigma^2_n$, $r_n = \rho_1 + \dots + \rho_n$ and denote by $F_n$ the cumulate distribution function (CDF) of $(x_1 + \dots + x_n)/s_n$. Then, for all $u$ and $n$, 
$$|F_n(u)-\Phi(u)| \leq \frac{6 r_n}{s^3_n},$$
where $\Phi$ is the standard normal CDF.
\end{lemma}
\begin{proof}
See \cite[p. 544, Th. 2]{Feller1966}
\end{proof}

\subsection{Proof of Theorem \ref{thm:Pe}}
\label{app:EPA-A}
We begin with the computation of $\prob(T_n>\tau_n)$ when $\tau_n>\Ex(T_n)$.  
Let $q_k$ be the PDF of $y_{k,s}$, the \emph{tilted} version of $y_k$ with PDF $p_k$, such that $q_k(y_{k,s}) = p_k(y_{k,s}) e^{s y_{k,s} -\mu_k(s)}$. It is straightforward to check that the first and second derivatives of $\mu_k(s)$ satisfy $\dot{\mu}_k(s)=\Ex(y_{k,s})$ and $\ddot{\mu}_k(s)=\Var(y_{k,s})$, respectively. Define $T_{n,s} = y_{1,s}+\dots+y_{n,s}$. The PDF of $T_{n,s}$ is 
\begin{align*}
Q_{T}(t_s) & =\int_{\mathcal{R}}\prod_{k=1}^n q_{k}(y_{k,s}) dy_{k,s}\\
 & =\int_{\mathcal{R}}\prod_{k=1}^n p_k(y_{k,s}) e^{s y_{k,s} -\mu_k(s)} dy_{k,s},
\end{align*}
where the integration domain $\mathcal{R}$ is defined as $\mathcal{R}=\{(y_1,\dots,y_n)\in\R^n: y_{1,s}+\dots+y_{n,s} = t_s\}$. Therefore,
\begin{align*}
Q_{T}(t_s) & =\left(\int_{\mathcal{R}}\prod_{k=1}^n p_k(y_{k,s}) dy_{k,s}\right) e^{s t_s -\sum_{k=1}^n\mu_k(s)} \\
& = P_T(t_s) e^{s t_s -\mu_T^n(s)},\\
\end{align*}
where $P_T$ is the PDF of $T_n$ and $\mu_T^n(s) = \sum_{k=1}^n\mu_k(s)$ is its LMGF. Thus, $T_{n,s}$ is the tilted version of $T_n$.
Then,
\begin{align}
\prob(T_n>\tau_n)\nonumber&= \int_{\tau_n}^{\infty} P_T(t) dt\nonumber\\
&= e^{-(s\tau_n - \mu_T^n(s))}\int_{\tau_n}^{\infty} Q_T(t) e^{-s(t-\tau_n)} dt \label{eq:step1}\\
&= e^{-(s\dot{\mu}_T^n(s) - \mu_T^n(s))}\int_{\dot{\mu}_T^n(s)}^{\infty} Q_T(t) e^{-s(t-\dot{\mu}_T^n(s))} dt. \label{eq:step2}
\end{align}
The first factor in (\ref{eq:step1}) is equal to the Chernoff bound if $s$ is constrained to be positive. 
Because the LMGF $\mu_T^n(s)$ is a convex function, $\dot{\mu}_T(s)$ is an nondecreasing function. By assumption, $\tau_n> \Ex(T_n)$. Thus, the value of $s$ that solves $\tau_n = \dot{\mu}_T^n(s)$ satisfies $s>0$ and we can use this equality to replace $\tau_n$ in (\ref{eq:step2}). 

The Chernoff error exponent defined as $\lim_{n\rightarrow\infty} \frac{1}{n}(s\dot{\mu}_T^n(s) - \mu_T^n(s))$ is asymptotically tight, i.e., it is equal to the error exponent $\lim_{n\rightarrow\infty}-\frac{1}{n}\log \prob(T_n>\tau_n)$. Therefore, we need to approximate the integral which is a sub-exponential factor. 
We let $F_T$ be the CDF of $Q_T$ and integrate by parts to obtain
\begin{align}
&\int_{\dot{\mu}_T^n(s)}^{\infty} Q_T(t) e^{-s(t-\dot{\mu}_T^n(s))} dt =\nonumber\\
&=\int_{\dot{\mu}_T^n(s)}^{\infty} s(F_T(t) - F_T\left(\dot{\mu}_T^n(s)\right) e^{-s(t-\dot{\mu}_T^n(s))} dt\nonumber\\
&=s\sqrt{ \ddot{\mu}_T^n(s)}\int_{0}^{\infty} (G_T(u) - G_T(0)) e^{-s \sqrt{ \ddot{\mu}_T^n(s)}\, u } du,\label{eq:QT}
\end{align}
where in the second equality we have defined the following zero-mean unit-variance random variable: $U_n=\frac{T_{n,s}- \dot{\mu}_T^n(s)}{ \sqrt{\ddot{\mu}_T^n(s)}} =\frac{\sum_k y_{k,s}-\dot{\mu}_k(s)}{ \sqrt{\ddot{\mu}_T^n(s)}}$ with CDF given by $G_T$. Applying Lem. \ref{lem:CLT} with $x_k\defeq y_{k,s}-\dot{\mu}_k(s)$, it can be shown that $s_n^2=\ddot{\mu}_T^n(s)=\mathcal{O}(n)$ and $r_n=\mathcal{O}(n)$. Therefore, $G_T(u) - G_T(0) = \Phi(u)-\Phi(0) + \mathcal{O}(\tfrac{1}{\sqrt{n}})$, where $\Phi$ is the standard Gaussian CDF.
The exponential factor in (\ref{eq:QT}) suggests that for sufficiently large $n$ we only need a good approximation of $\Phi(u)-\Phi(0)$ around 0, and only within a small fraction of the standard deviation of $T_n$, $\sqrt{ \ddot{\mu}_T^n(s)}$. 
Thus, it is sufficient to use the following bounds: $1-\tfrac{u^2}{2}\leq e^{-\frac{u^2}{2}}\leq 1$ which produces, if $u\geq 0$, $\frac{u}{\sqrt{2\pi}}-\frac{u^3}{6\sqrt{2\pi}} \leq \Phi(u)-\Phi(0)\leq \frac{u}{\sqrt{2\pi}}$. Then, it easy to show that
\begin{align*}
\int_{\dot{\mu}_T^n(s)}^{\infty} Q_T(t) e^{-s(t-\dot{\mu}_T^n(s))} dt = \frac{1}{\sqrt{2\pi s^2\ddot{\mu}_T^n(s))}} + \mathcal{O}\left(\frac{1}{\sqrt{n}}\right).
\end{align*}
Replacing this in (\ref{eq:step2}) produces (\ref{eq:generalPfa}), the first result of the theorem. The computation of $\prob(T_n<\tau_n)$ is similar and it is omitted.

\subsection{Proof of Corollary \ref{cor:GQF}}
\label{app:EPA-B}
In order to apply Thm. \ref{thm:Pe}, the Gaussian quadratic form (\ref{eq:GaussQF}) can be expressed as the sum of $n'$ independent variables.
Then, we only need to evaluate its LMGF $\mu_i(s)= \log\Ex_i\left(e^{sT_{n'}}\right)$ under $\Hip_i$, $i=0,1$, and their derivatives:
\begin{align}
\mu_i(s) & = -\log\det\left(I_n - s \Xi_{i,n'}P_{n'} \right),\label{eq:LMGF-GaussQF}\\ 
\dot{\mu}_i(s)&=\tr\{[(\Xi_{i,n'}P_{n'})^{-1}-sI_{n'}]^{-1}\},\label{eq:LMGF-GaussQF-dot}\\ 
\ddot{\mu}_i(s)&=\tr\{[(\Xi_{i,n'}P_{n'})^{-1}-sI_{n'}]^{-2}\}.\label{eq:LMGF-GaussQF-ddot}
\end{align}
Replacing these expressions in (\ref{eq:generalPfa}) and (\ref{eq:generalPm}) completes the proof. 

\section{Performance for the equicorrelated process: proof of Lemma \ref{lem:FAP}}
\label{app:FACS}
\subsection{$P_\text{fa}$ for the LCS strategy}
We use Th. 1 to compute $P_\text{fa,LCS}$. We first evaluate the LMGF of the statistic $T_\text{LCS}$ in (\ref{eq:LCS-2}) under $\Hip_0$ and its first and second derivatives given by (\ref{eq:LMGF-GaussQF})-(\ref{eq:LMGF-GaussQF-ddot}) in Corollary \ref{cor:GQF}, where $P_{n'}=\Gamma_{n'}$.

The covariance matrix of the equicorrelated process is a circulant matrix and, therefore, it is diagonalized by the $l\times l$ DFT matrix $F_{ll}$ and its eigenvalues are the DFT of its first column, i.e., $\Sigma_{s,l}=F_{ll} D_{s,l} F_{ll}^H$ where $D_{s,l}=\diag(1+(l-1)\rho,1-\rho,\dots,1-\rho)$.
Considering the PSD of the equicorrelated process, the frequency vectors used in the precoding matrix of the PFS-MAC scheme given in Def. \ref{def:PFS-MAC} have the indexes $j_{k'}=k'-1$, $k'=1,\dots,l'$. Then, the covariance matrix in (\ref{eq:cov_zlp1}) is
\begin{equation*}
\Xi_{1,l'}= \gamma_0^2 D_{s,l'} + (\gamma_0^2\sigma^2_v + \sigma^2_w) I_{l'},
\end{equation*}
where $D_{s,l'}$ is a $l'\times l'$ diagonal matrix, $D_{s,l'}=\diag(1+(l-1)\rho,1-\rho,\dots,1-\rho)$. Notice that $\Xi_{1,l'}$ is a diagonal matrix because the covariance of the process is a circulant matrix and the precoding scheme is based on the DFT matrix.
The covariance matrices of $\ve z$ under $\Hip_0$ and $\Hip_1$ are given by (\ref{eq:Xi-np}) and they result  
\begin{align}
\Xi_{0,n'} &= (\gamma_0^2\sigma_v^2+ \sigma^2_w) I_{n'},\nonumber\\
\Xi_{1,n'} &=\gamma_0^2\sigma^2_s(1-\rho)I_{n'} + \gamma_0^2\sigma^2_s\rho (I_m\otimes F_{ll'}^H)\ve{1_l} (I_m\otimes F_{ll'}) \nonumber \\
&+ (\gamma_0^2\sigma^2_v + \sigma^2_w) I_{n'}, \nonumber\\
&=\gamma_0^2\sigma^2_s ( (1-\rho)I_{n'} + \rho \ve{1}_m \otimes E_{l'}) + (\gamma_0^2\sigma^2_v + \sigma^2_w) I_{n'}, \nonumber
\end{align}
where $\otimes$ is the Kronecker product and the $l'\times l'$ matrix $E_{l'}$ is defined as $E_{l'} = \diag(l,0,\dots,0)$. Then, 
\begin{align*}
\Xi_{0,n'}\Gamma_{n'} & = I_m\otimes\diag\left(\tfrac{\gamma_0^2 d_1}{\gamma_0^2(d_1 + \sigma_v^2) + \sigma_w^2},\dots,\tfrac{\gamma_0^2 d_{l'}}{\gamma_0^2(d_{l'} + \sigma_v^2) + \sigma_w^2}\right),
\end{align*}
where $d_1 = \sigma^2_s (1+(l-1)\rho)$ and $d_k = \sigma^2_s (1-\rho)$ if $k=2,\dots,l'$. Using (\ref{eq:LMGF-GaussQF})-(\ref{eq:LMGF-GaussQF-ddot}) of Corollary \ref{cor:GQF}, and denoting $\Gamma=\snr_\text{PFS-MAC}(l)$, the LMGF of the statistic (\ref{eq:LCS-2}) under $\Hip_0$ and its first and second derivatives are: 
\begin{align}
\mu_0(s_0) &= - \tfrac{n}{l}\log\left(1-s_0 \tfrac{\Gamma(1+(l-1)\rho)}{1+\Gamma(1+(l-1)\rho)}\right) \nonumber\\
&-n\left(\beta-\tfrac{1}{l}\right)\log\left(1-s_0 \tfrac{\Gamma(1-\rho)}{1+\Gamma(1-\rho)}\right)\label{eq:mu0},\\
\dot{\mu}_0(s_0) &= \tfrac{n}{l} \tfrac{\Gamma(1+(l-1)\rho)}{1+\Gamma(1+(l-1)\rho)-s_0 \Gamma(1+(l-1)\rho)}\nonumber\\
&+ n\left(\beta-\tfrac{1}{l}\right) \tfrac{\Gamma(1-\rho)}{1+\Gamma(1-\rho)-s_0 {\Gamma(1-\rho)}},\label{eq:mu0dot}\\
\ddot{\mu}_0(s_0) &= \tfrac{n}{l} \tfrac{(\Gamma(1+(l-1)\rho))^2}{(1+\Gamma(1+(l-1)\rho)-s_0 \Gamma(1+(l-1)\rho))^2}\nonumber\\
&+ n\left(\beta-\tfrac{1}{l}\right) \tfrac{(\Gamma(1-\rho))^2}{(1+\Gamma(1-\rho)-s_0 {\Gamma(1-\rho)})^2}.\label{eq:mu0ddot}
\end{align}
Now, the mean of the statistic $T_\text{LCS}$ under $\Hip_1$ is 
$m_1 = \Ex_1(T_\text{LCS})= n\Gamma (\rho + \beta(1-\rho))$.
Then, $\mu_0(s_0^*)=\tau_n = m_1$ implies that $s_0^*=1$. Replacing it in (\ref{eq:mu0})-(\ref{eq:mu0ddot}) and using Th. \ref{thm:Pe}, we obtain (\ref{eq:PfaEqui-LCS}).

\subsection{$P_\text{fa}$ for the FCS strategy}
The statistic $T_\text{FCS}$ is an LLR, for which the LMGFs under $\Hip_0$ and $\Hip_1$ satisfy 
\begin{equation}
\mu_0(s)=\mu_1(s-1),\label{eq:mu0mu1Equiv}
\end{equation} where $\mu_i(s)=\Ex_i(e^{sT_\text{FCS}})$. The direct computation of $\mu_0(s)$ implies to obtain firstly the inverse of $\Xi_{1,n'}$, which is not straightforward. An easier way is to compute $\mu_1(s)$ for which we only need to invert the diagonal matrix $\Xi_{0,n'}$ and then, to use (\ref{eq:mu0mu1Equiv}) to evaluate $\mu_0(s)$. The expression for $\mu_1(s)$ and its first and second derivatives are:
\small
\begin{align}
&\mu_1(s_1)  =\nonumber\\
& -s_1\log\left(1+ \Gamma(1+(n-1)\rho)\right) -s_1(n\beta-1)\log\left(1+ \Gamma(1-\rho)\right)\nonumber\\
&-\log\left(1-s_1\Gamma(1+(n-1)\rho)\right) -\left(n\beta-1\right)\log\left(1-s_1\Gamma(1-\rho)\right),\nonumber\\
&\dot{\mu}_1(s_1) = \nonumber\\ 
& -\log\left(1+ \Gamma(1+(n-1)\rho)\right) -(n\beta-1)\log\left(1+ \Gamma(1-\rho)\right)\nonumber\\
&+\frac{\Gamma(1+(n-1)\rho)}{1-s_1\Gamma(1+(n-1)\rho)} +\left(n\beta-1\right)\frac{\Gamma(1-\rho)}{1-s_1\Gamma(1-\rho)},\nonumber\\
&\ddot{\mu}_1(s_1)  = \nonumber\\
&\frac{(\Gamma(1+(n-1)\rho))^2}{(1-s_1\Gamma(1+(n-1)\rho))^2} +\left(n\beta-1\right)\frac{(\Gamma(1-\rho))^2}{(1-s_1\Gamma(1-\rho))^2}.\nonumber
\end{align}
\normalsize
$\tau_\text{FCS}=\Ex_1(T_\text{FCS})=\dot{\mu}_1(s_1)=\dot{\mu}_0(s_0)$ implies that $s_1=0$ and $s_0=1$.  Using this in (\ref{eq:mu0mu1Equiv}) and replacing all in (\ref{eq:generalPfa}), we obtain (\ref{eq:PfaEqui-FCS}).
Notice that this result coincides with the one in (\ref{eq:PfaEqui-LCS}) if $l$ is replaced by $n$. The difference between both FCS and LCS strategies is that the first one contemplates the full correlation of the process. 

\subsection{$P_\text{fa}$ for the ED strategy}
The energy detector $T_\text{ED}$ can be normalized by $\gamma_0^2\sigma_v^2 + \sigma_w^2$ without modifying the performance: $\tilde{T}_\text{ED}=\frac{T_\text{ED}}{\gamma_0^2\sigma_v^2 + \sigma_w^2}$. Its LMGF under $\Hip_0$ and first and second derivatives are:
\begin{align}
\mu_0(s_0) &= -n\log(1-s_0),\label{eq:muE}\\
\dot{\mu}_0(s_0) &= \frac{n}{1-s_0},\label{eq:muEdot}\\
\ddot{\mu}_0(s_0) &= \frac{n}{(1-s_0)^2}\label{eq:muEddot}.
\end{align}
$\tau_\text{ED}=\Ex_1(\tilde{T}_\text{ED})=n(1+\Gamma)=\dot{\mu}_0(s_0)$ implies that $s_0=\frac{\Gamma}{1+\Gamma}$. Using this in (\ref{eq:muE})-(\ref{eq:muEddot}) and replacing all in (\ref{eq:generalPfa}) results in (\ref{eq:PfaEqui-ED}). 
Notice that the same result can be obtained by replacing $l$ by $n$ in $(\ref{eq:PfaEqui-LCS})$ and taking $\rho=0$, given that the ED strategy does not consider the  correlation of the process.

\bibliographystyle{IEEEtran}
\bibliography{IEEEabrv,../Refs/refs}
\end{document}